\documentclass[journal]{IEEEtran}
\ifCLASSINFOpdf
\else
\fi

\usepackage{cite}
\usepackage{amsmath,amssymb,amsfonts}
\usepackage{algorithmic}
\usepackage{graphicx}
\usepackage{textcomp}
\usepackage{xcolor}
\usepackage{stfloats}
\usepackage{amsmath}
\usepackage{algorithm}
\usepackage{algorithmic}

\usepackage{multirow}  
\usepackage{array}	   
\usepackage{booktabs}  
\usepackage{tabularx}  
\usepackage{longtable} 
\usepackage{tabu}      
\usepackage{threeparttable} 
\usepackage{caption}   
\usepackage{graphicx}  
\usepackage{caption}
\usepackage{subfigure} 
\usepackage{url}
\usepackage{epstopdf}

\newtheorem{proposition}{Proposition}

\newtheorem{proof}{Proof}

\begin{document}
%
\title{Digital Twin-Assisted Federated Learning with Blockchain in Multi-tier Computing Systems}
%
%
%

\author{Yongyi~Tang,~Kunlun~Wang,~Dusit~Niyato,~Wen Chen,~and~George~K.~Karagiannidis
\thanks{Yongyi Tang and Kunlun Wang are with the School of Communication and Electronic Engineering, East China Normal University, Shanghai 200241, China (e-mail: 51255904070@stu.ecnu.edu.cn; klwang@cee.ecnu.edu.cn).}
\thanks{Dusit Niyato is with the School of Computer Science and Engineering, Nanyang Technological University, Singapore 639798 (e-mail: dniyato@ntu.edu.sg).}
\thanks{Wen Chen is with the Department of Electronic Engineering, Shanghai Jiao
Tong University, Shanghai 200240, China (email: wenchen@sjtu.edu.cn).}
\thanks{George K. Karagiannidis is with the Department of Electrical and Computer Engineering, Aristotle University of Thessaloniki, 54 124 Thessaloniki, Greece and also with the Cyber Security Systems and Applied AI Research Center, Lebanese American University (LAU), Lebanon (e-mail: geokarag@auth.gr).}
}

\maketitle

\begin{abstract}
In Industry 4.0 systems, a considerable number of resource-constrained Industrial Internet of Things (IIoT) devices engage in frequent data interactions due to the necessity for model training, which gives rise to concerns pertaining to security and privacy. In order to address these challenges, this paper considers a digital twin (DT) and blockchain-assisted federated learning (FL) scheme. To facilitate the FL process, we initially employ fog devices with abundant computational capabilities to generate DT for resource-constrained edge devices, thereby aiding them in local training. Subsequently, we formulate an FL delay minimization problem for FL, which considers both of model transmission time and synchronization time, also incorporates cooperative jamming to ensure secure synchronization of DT. To address this non-convex optimization problem, we propose a decomposition algorithm. In particular, we introduce upper limits on the local device training delay and the effects of aggregation jamming as auxiliary variables, thereby transforming the problem into a convex optimization problem that can be decomposed for independent solution. Finally, a blockchain verification mechanism is employed to guarantee the integrity of the model uploading throughout the FL process and the identities of the participants. The final global model is obtained from the verified local and global models within the blockchain through the application of deep learning techniques. The efficacy of our proposed cooperative interference-based FL process has been verified through numerical analysis, which demonstrates that the integrated DT blockchain-assisted FL scheme significantly outperforms the benchmark schemes in terms of execution time, block optimization, and accuracy.
\end{abstract}


%
\IEEEpeerreviewmaketitle

\section{Introduction}
%
%
%
%
\IEEEPARstart{T}{he} deployment of fifth-generation (5G) wireless networks and the advancement of sixth-generation (6G) wireless networks have prompted the industry to explore relevant technologies, requirements, and use cases for Industry 4.0. Industry 4.0 relies on advanced machine capabilities and accelerated data analytics combined with artificial intelligence (AI) to build autonomous, self-configuring systems that optimize manufacturing efficiency, precision, and accuracy through the integration of new methods, including the Internet of Things (IoT), digital twin (DT), federated learning (FL), and blockchain \cite{10234720}, \cite{10114989}. However, Industry 4.0 requires distributed intelligent services to adapt to dynamic environments in real time. Due to the complexity of industrial environments and the heterogeneity of Industrial Internet of Things (IIoT) devices, ensuring the security and privacy of data collection and processing among various participants in the industrial ecosystem is a challenging task.
\par FL is a distributed collaborative model training method that emphasizes privacy protection and enables disparate devices to collaboratively build accurate and stable global models. Compared to traditional centralized learning methods, FL enables more efficient data processing by relying on collaboration among various participating nodes to achieve distributed training. This decentralizes the training process as each device is responsible for processing a portion of the training data, allowing for faster and more accurate analysis \cite{9692911}, \cite{9815106}.
\par At the same time, DT-based FL has attracted considerable interest in the context of IIoT. DT facilitates the transformation of physical entities or systems in industrial contexts into their digital forms, enabling the modeling of industrial ecosystems, real-time monitoring, prediction, and interaction within virtual environments. In essence, DT bridges the gap between the physical and virtual realms, facilitating data collection and simulation of industrial processes. The use of DTs has enabled the migration of real-time data analysis and processing to the edge, enhancing the effectiveness of machine learning algorithms by allowing distributed learning solutions to be deployed in intelligent industrial environments \cite{10091890}, \cite{10188847}, \cite{9899718}.
\begin{figure}[htbp]
    \centering
    \includegraphics[width=7.5cm]{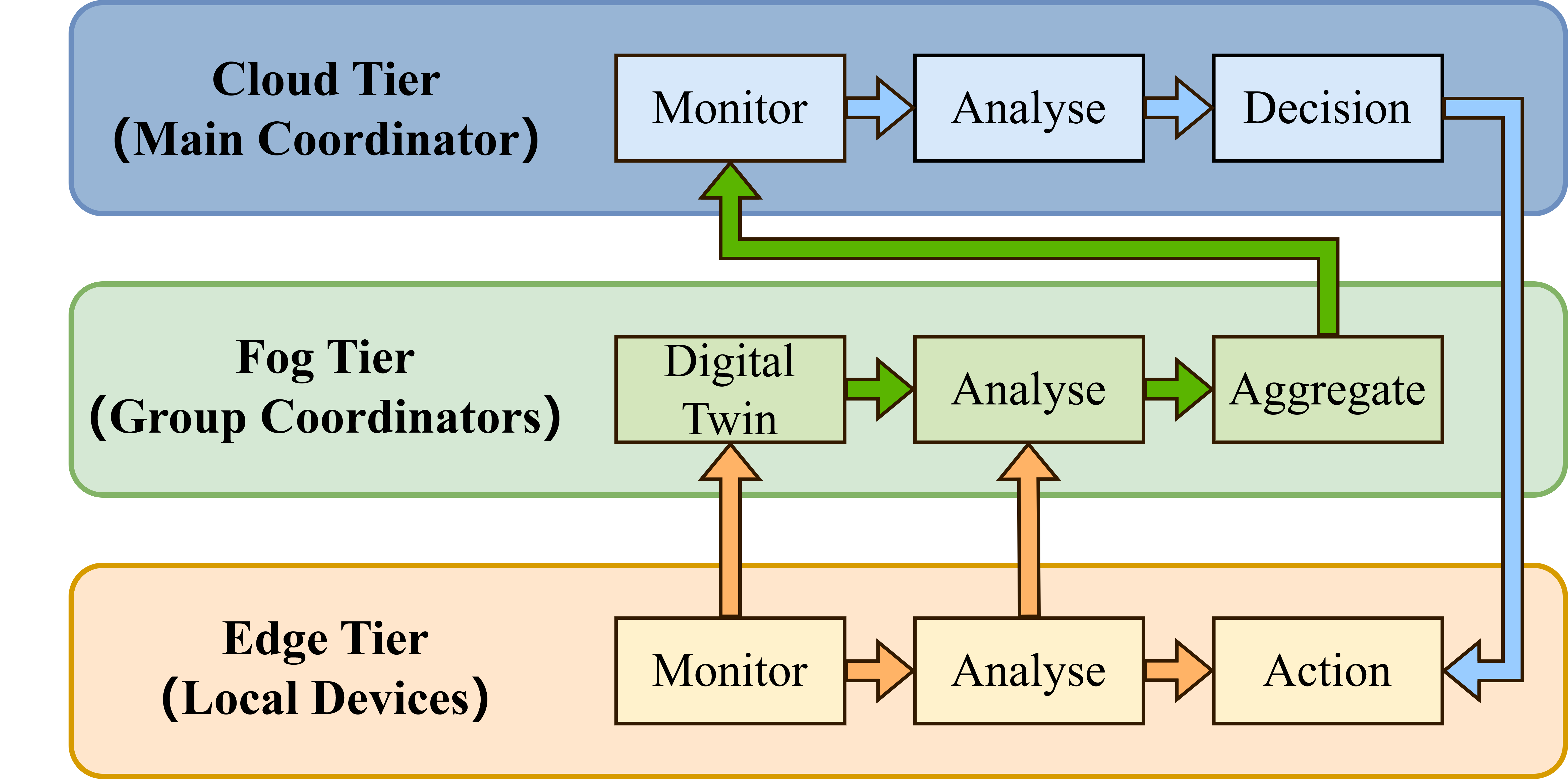}
    \caption{A multi-tier collaboration model of equipment in Industry 4.0.}
    \label{P1}
\end{figure}
\par In Fig. \ref{P1}, we illustrate the connections and interdependencies between industrial devices, edge devices, and main coordinator devices at different levels of the Industry 4.0 decision-making process. Establishing such cooperation and correlation is essential for achieving and continuously improving system efficiency \cite{9978919}. Local devices and their corresponding DTs perform continuous monitoring, assessing the status of both devices and the surrounding environment through data collection. The collected industrial data is used for local training and learning processes related to the equipment. In addition, the creation of DTs is achieved by sharing physical behaviors and state characteristics between local devices and edge devices. DTs can coordinate the training process and even replace local or edge devices in performing training tasks, thereby enriching the data sources for FL and improving the quality of FL training.
\par However, due to the diversity of participating devices and environments, certain limitations exist, including resource capability, communication efficiency, security, and other performance-related metrics \cite{9709603}. For example, devices with limited resources may only be able to perform data collection and upload tasks, lacking the capacity for local training and secure data transfer. Although these local devices can use the resources of their respective DTs and edge servers to perform local training tasks, during the synchronization process of DTs with local device attribute information or dataset updates, malicious attackers can intercept signals from industrial devices due to the open access nature of wireless channels, potentially leading to eavesdropping or data transmission tampering. To address these security concerns, the academic community has proposed a novel approach using artificial jamming, in which friendly jammers deliberately transmit jamming signals to interfere with eavesdroppers' ability to receive and decode transmitted data. This strategy has been identified as a promising solution for improving data security throughput \cite{9839650}, \cite{9580594}. Specifically, cooperative jamming among wireless devices improves data transmission security and reduces the delay of FL iterations. When local devices upload data, wireless devices simultaneously transmit jamming signals to eavesdroppers, improving secure data throughput. However, to maximize the benefits of cooperative jamming while managing the energy constraints of wireless devices, it is critical to accurately allocate each device's energy budget and ensure an optimal allocation between local model training, data transmission, and cooperative jamming.
\par While the introduction of FL and DT technologies in IIoT offers advantages in privacy protection and distributed model training, challenges remain, such as the lack of trust mechanisms, low collaboration efficiency, scalability issues, and model verification challenges. The use of blockchain in industrial scenarios enhances the effectiveness and reliability of FL. The decentralized and tamper-proof nature of blockchain facilitates the recording of user access and ensures data integrity. Blockchain operates without relying on a central server or single entity, and is collectively maintained by multiple nodes within the network \cite{9676337}. This architecture allows heterogeneous devices to interact with updated information and rely on global records for accurate decision making, which is critical for fully autonomous, adaptive, and self-healing industrial systems. In addition, the automatic execution of smart contracts reduces human intervention, thereby improving productivity and quality. However, as industry advances, the amount of data that needs to be processed is growing exponentially, creating scalability challenges for blockchain's performance and adaptability. These challenges can lead to increased blockchain latency and reduced throughput, jeopardizing the system's accuracy, security, and privacy.
\par In this article, we focus on designing a blockchain- and DT-enabled FL solution to support industrial model training services in Industry 4.0 scenarios. Specifically, we use DT to enable resource-limited industrial devices to participate in FL, enrich the data sources for FL, and improve the accuracy of model training. We then address the communication security issues related to the synchronization of DTs for resource-constrained industrial devices, and propose an approach that enhances the reliability of data transmission through cooperative jamming while minimizing the delay required for FL. Finally, by using blockchain and introducing a proposed validator selection algorithm, we further provide trust to FL participants and ensure the integrity of shared models. Moreover, within our proposed FL framework, the blockchain can achieve block optimization, reducing both the size and number of blocks. Below is a summary of our contributions to this article:
\par 1) We propose a blockchain and DT-assisted FL solution for Industry 4.0. The proposed solution allows local industrial devices with limited computational resources to participate in FL through DT.
\par 2) We introduce cooperative jamming for industrial local devices with limited resources to ensure the secure synchronization of corresponding DTs. Furthermore, we propose an effective algorithm for solving joint optimization problems to reduce the delay in FL iterations.
\par 3) We use blockchain to validate and verify uploaded local/global models, ensuring maximum data privacy and integrity.

\par The remainder of this article is organized as follows. Section \uppercase\expandafter{\romannumeral2} presents a review of the relevant literature. In Section \uppercase\expandafter{\romannumeral3}, we describe the blockchain and DT-assisted FL solution, and formulate the FL delay optimization problem within industrial equipment clusters. Section \uppercase\expandafter{\romannumeral4} introduces an algorithm for solving joint optimization problems. In Section \uppercase\expandafter{\romannumeral5}, we conduct a comparative analysis of the simulation results against other benchmark schemes. Finally, Section \uppercase\expandafter{\romannumeral6} summarizes the work presented and outlines potential avenues for future research.


\section{Related Work}
This section provides a comprehensive review of the applications of FL utilizing DT and blockchain technology in industrial contexts. Additionally, it offers an overview of pertinent literature that incorporates artificial jamming techniques to enhance security throughput while minimizing FL iteration latency.

\subsection{FL-Supported IIoT}
FL has gained significant traction in various applications, particularly in IIoT and edge computing environments. Traditionally, AI capabilities have been hosted in cloud or data center infrastructures, which limits the rapid growth of IIoT data volumes. Transferring large volumes of IIoT data to remote servers for model training requires high network bandwidth and creates significant communication overhead, making it unsuitable for time-sensitive applications. FL addresses this by averaging local updates from multiple clients without accessing their data, reducing data privacy risks. Furthermore, FL leverages the computational resources of numerous IIoT devices, improving the quality of data training compared to centralized machine learning methods. In addition, implementing FL in industrial environments can improve data offloading and caching, provide collaborative intelligence, detect and prevent security threats, and support mobile crowdsensing \cite{9508907}.
\par In \cite{10473157}, the authors propose a security-conscious computation offloading methodology that utilizes federated reinforcement learning in the context of IIoT. The computational tasks in smart factories are represented by a directed acyclic graph, and the offloading problem is formulated as a Markov decision process. This approach considers the optimization of latency, energy consumption, and the number of overdue tasks, and employs differential privacy to ensure data security. Deep reinforcement learning techniques are used to derive near-optimal offloading decisions. In \cite{9130128}, the authors present a joint attack detection and defense solution for IIoT systems that leverages the privacy-enhancing capabilities of FL. Each IIoT device is equipped with a locally operable deep neural network, allowing it to retrain its threat model to effectively counter adversarial attacks. Much of the existing research on FL has not adequately addressed the challenges associated with edge device connectivity and resource allocation in industrial scenarios. This oversight is largely due to the neglect of heterogeneity resulting from data interactions and the specific domains of traditional industrial devices.

\subsection{Digital Twin-Assisted FL}
\par DT-Assisted FL provides an innovative solution for the advancement of Industry 4.0. In \cite{10133837}, the authors presented a DT-assisted architecture for industrial mobile crowdsensing based on an incentive-driven FL framework, complemented by a contract-based reputation mechanism and a Stackelberg-based customer incentive mechanism. This approach aims to achieve dynamic sensing in complex IIoT environments, characterized by heterogeneous and resource-constrained mobile clients. In \cite{9244624}, the authors propose a DT-supported architecture for IIoT to capture the dynamic characteristics of industrial devices, thereby accelerating FL convergence, promoting collaborative learning, and enhancing learning efficiency. They also introduce a trustworthy aggregation-based FL approach to address potential estimation biases between DTs and actual device state values, thereby addressing the heterogeneity of IIoT environments. However, in terms of DT-enabled industrial device participation in FL, offloading all operational data to the DT may be impractical due to resource constraints. This method would incur significant communication costs, resource consumption, and time delays, in addition to raising privacy concerns.
However, with respect to the participation of DT-supported industrial equipment in the FL, offloading all operational data to the DT may be impractical due to resource constraints. This approach could result in significant communication costs, resource consumption and transmission delays. In addition, data transmission between the DT and industrial equipment is primarily via wireless links, which, due to the open nature of wireless channels, can lead to privacy and security issues \cite{9854182}.

\subsection{Secure Synchronization of DT}
To improve the security of DT synchronization, some studies have explored cryptographic schemes \cite{9765576}. In \cite{10234580}, the authors proposed an unbounded and efficient directly revocable attribute-based encryption scheme with adaptive security tailored for DTs. By using arithmetic span programs as access structures, this scheme effectively achieves revocability and fine-grained access control. In \cite{10239228}, the authors leveraged the advantages of hierarchical encryption schemes and privacy-preserving rollback re-encryption schemes to propose a blockchain-aware rollback data access control solution that enables dynamic access control to DT data while maintaining the immutability of the blockchain. Implementing secure communication protocols for IIoT systems that incorporate DTs is essential for maintaining effective network operations and ensuring secure data transmission. However, it is important to note that this encryption method may not be computationally efficient when applied to large-scale machine learning models, as industrial devices typically lack the processing power required to perform encryption and decryption operations.
\par In \cite{10566993}, the authors proposed a DT communication-friendly jamming method based on deep reinforcement learning. This approach optimizes the interference frequency, power, and duration by using friendly jammers, providing an effective means to combat active eavesdropping. By implementing a secure layered architecture, the anti-eavesdropping performance and confidentiality rate can be improved by leveraging information about the channel status between devices and servers, the malicious interference intensity of active eavesdroppers, and details about eavesdropping channels. While the use of third-party friendly jammers can facilitate secure data transmission for industrial devices, the associated costs may not be viable for industrial scenarios that prioritize cost-effectiveness and high productivity. In \cite{9738843}, the authors use FL to enable cooperative jamming by local devices, thereby mitigating the risk of eavesdropping during global model broadcasts from the central server. In addition, they introduce a hierarchical algorithm to reduce the iteration delay inherent in FL. However, further research is needed to improve the security of uplink communication in the context of FL in industrial applications.

\begin{figure*}[htbp]
    \centering
    \includegraphics[width=11cm]{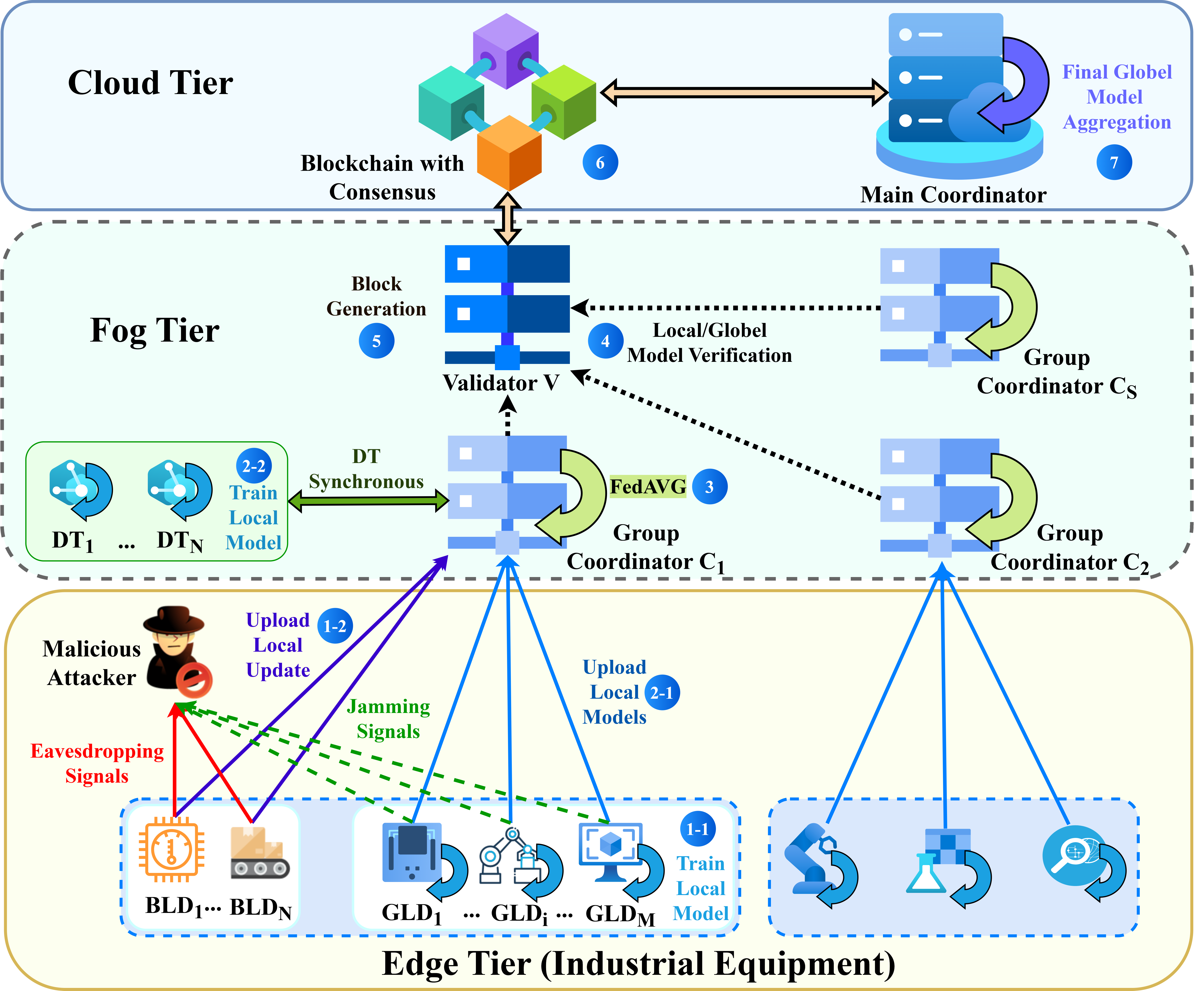}
    \caption{The framework of DT-assisted FL with blockchain in multi-tier computing systems.}
    \label{P2}
\end{figure*}

\subsection{Blockchain-based FL in IIoT}
The blockchain-based FL scheme for sharing distributed data among multiple untrusted parties is particularly well suited for applications within the IIoT. In \cite{8843900}, the authors integrate FL into the consensus process of a licensed blockchain, effectively preserving data privacy by sharing data models rather than disclosing the actual data. In \cite{9134967}, the authors propose a decentralized paradigm for big data-driven cognitive computing that uses FL to address the challenges posed by data silos. The integration of blockchain provides incentive mechanisms that facilitate full decentralization and enhance resilience to poisoning attacks in industrial automation processes. Furthermore, in \cite{9714771}, the authors present an average FL mechanism based on homomorphic encryption coupled with a credit data storage mechanism within the blockchain. This model allows credit data providers to effectively monitor the data usage process.

\par In contrast to the aforementioned studies, the proposed blockchain and DT-assisted FL framework effectively supports industrial model training services in Industry 4.0 scenarios. Specifically, we leverage DT technology to enable resource-constrained industrial devices to participate in federated learning, thereby enriching the data sources and improving the accuracy of model training. In addition, we address the communication security issues faced by resource-constrained industrial devices during the synchronization of DTs by proposing a cooperative jamming approach to improve the reliability of data transmission while minimizing the delay required for Fl. Finally, through the use of blockchain and our proposed validator selection algorithm, we further enhance trust among FL participants and ensure the integrity of shared models. In addition, within our proposed FL framework, blockchain can achieve block optimization, reducing both the size and number of blocks. The definitions of some common notations used in this paper are summarized in Table \ref{tb1}.

\begin{table}[!htbp]
\centering
\caption{Common Notations Used}
\renewcommand{\arraystretch}{1.25} 
\begin{tabular}{c|c}
    \toprule
    Notation & Definition \\
    \midrule
    $t_i^{\rm loc}$ & Latency for $GLD_i$ to complete local training \\
    $t^{\rm up}_{\rm G}$ & Duration of the $\{GLD_i\}_{i \in\mathcal{M}}$ uploading transmission \\
    $t^{\rm loc}_{\rm B}$ & Latency for local training of the DTs \\
    $t^{\rm up}_{\rm B}$ & Duration of the $\{BLD_j\}_{j \in\mathcal{N}}$ uploading transmission \\
    $E_i^{\rm loc}$ & Energy consumption of $GLD_i$ to complete local training \\
    $E_i^{\rm up}$ & Energy consumption for $GLD_i$ to upload local model \\
    $E^{\rm jam}_i$ & Energy consumption of $GLD_i$ on jamming signal \\
    $v_i$ & CPU rate of $GLD_i$ \\
    $D_{i,\rm G}$ & Size of $GLD_i$ local training data \\
    $D_{j,\rm B}$ & Size of $BLD_j$ local training data \\
    $p_i^{\rm G}$ & Transmission power of $GLD_i$ \\
    $p_j^{\rm B}$ & Transmission power of $BLD_j$ \\ 
    $R_j^{\rm sec}$ & Secrecy rate of $BLD_j$ \\
    $y$ & Upper bound of all GLDs’ local training latency \\
    $Q$ & Aggregate jamming effect \\
    \bottomrule
\end{tabular}
\label{tb1}
\end{table}

\section{System Framework and Problem Formulation}
We demonstrate the complete framework of DT-assisted FL with blockchain in Fig. \ref{P2}. We consider multiple industrial clusters managed by their respective group coordinators $\{C_s\}_{s\in S}$. Within an industrial cluster, there are $k$ local IIoT devices, indexed by $\{LD_1, LD_2,\cdots, LD_k\}$. After classifying the local devices, we define $\mathcal{M}$ devices with sufficient resources as $\{GLD_i\}_{i\in \mathcal{M}}$ and $\mathcal{N}$ devices with limited resources as $\{BLD_j\}_{i \in \mathcal{N}}$. $\{GLD_i\}_{i\in \mathcal{M}}$ can send cooperative jamming signals to malicious attackers to ensure the security of DT synchronization. The group coordinator can generate $\mathcal{N}$ DTs for $\{BLD_j\}_{i \in \mathcal{N}}$ to facilitate local training, and perform averaging aggregation of the local models within the cluster to produce local/global models. We use a validator selection algorithm to determine validators $V$ for authenticating participants and verifying model integrity, thereby generating blocks that can be added to the blockchain. Finally, the main coordinator aggregates all verified local/global models using deep learning techniques to obtain the final global model \cite{10234720}.

\subsection{DT-Assisted Industrial Cyber-Physical Systems}
Considering that not all IIoT devices have the capability to perform the computationally intensive tasks required for the FL process, local devices can offload FL training tasks to other edge devices with more computational resources. Furthermore, to optimize the use of local device data and mitigate the risk of data leakage or tampering associated with the transmission of raw data, the FL process used in this work leverages the DT of local industrial devices with limited resources to train the local device model. The edge nodes create and update the DT of the industrial devices based on relevant data, and then update the blockchain with this information. Any changes in the status of the device will be added to the blockchain in the form of blocks. The DT modeling of local industrial equipment with insufficient resources can be represented as follows:

\begin{equation}\label{eq1}
\begin{aligned}
DT_{BLD_j}(t)&=\{State_{BLD_j}(t), Static_{BLD_j}(t),\\& Resource_{BLD_j}(t), Deviation_{BLD_j}(t)\},
\end{aligned}
\end{equation}
where $State_{BLD_j}(t)$ represents the training state of the mapped device $BLD_j$ during the FL process at time $t$, while $Static_{BLD_j}(t)$ represents the static operating data of the device $BLD_j$ at time $t$. $Resource_{BLD_j}(t)$ represents the resource state of $BLD_j$ at time $t$, including computing power, communication bandwidth, memory capacity, and power supply capabilities. In general, there may be a mismatch between the mapping value of DT at time $t$ and the mapping value of the local device capability characteristics. This can be defined as the deviation $Deviation_{BLD_j}(t)$ between the actual value and the mapping DT value. In fact, upon receiving the updated device state, DT performs self-calibration to maintain the minimum deviation, thus ensuring dynamic optimization through accurate decision making. It should be noted that although (\ref{eq1}) indicates a bias in the mapping of DT, this bias can occur over an extended period of time and can be considered negligible when performing training tasks on DT \cite{9244624}.

\subsection{Blockchain with Conscensus}
In our framework, a consortium blockchain is established between the group coordinator and the master coordinator. This blockchain serves all clusters and edge devices. Before adding blocks, validator nodes verify all blocks to be added. Once the group coordinator completes the model update, the validators create a block. The validation process focuses on verifying the identities of the participants and ensuring the integrity of the model against the primary global model. Once verified, the block contains all collected records and is shared with all participants who wish to be added to the blockchain. For this purpose, it is essential to select fog devices as validators. The number of validators assigned to the fog tier is determined by the main coordinator. The validator selection process uses a reputation-based mechanism \cite{9737334}. According to (\ref{eq2}), the reputation score of the group coordinator can be calculated. This is achieved by combining a number of factors, including the coordinator's ability to encrypt data, its ability to route securely, the proportion of supporting local devices that generate DTs, and its historical behavioral characteristics. The reputation score is expressed as

\begin{equation}\label{eq2}
\begin{aligned}
Score^{\rm Rep}_{C_s}&=\omega_{Enc}(Cap^{\rm Enc}_{C_s}(t))+\omega_{Rou}(Cap^{\rm Rou}_{C_s}(t))\\&+\omega_{Prop}(Prop_{C_s}(t))+\omega_{Hist}(Hist_{C_s}(t)),
\end{aligned}
\end{equation}
where $\omega_{Enc}$, $\omega_{Rou}$, $\omega_{Prop}$ and $\omega_{Hist}$ represent the score weight. $Cap^{\rm Enc}_{C_s}(t)$ represents the encryption capability of the group coordinator. Assuming that all group coordinators need to encrypt publicly accessible data, a shorter encryption time and higher randomness in the encrypted data will result in a higher score for the group coordinator. $Cap^{\rm Rou}_{C_s}(t)$ denotes the secure routing capability of a device, with the assumption that the data to be transmitted is public. A group coordinator with direct communication capability to the main coordinator, for instance, without the use of a relay or access point, will achieve a higher score. $Prop_{C_s}(t)$ is the proportion of $LD_k$ that the group coordinator assists $BLD_j$ in generating DT for local training. A higher proportion increases the risk of $BLD_j$ synchronization information being tampered with, thereby affecting the accuracy of the final model. Consequently, the fewer DTs the group coordinator needs to generate, the higher the score will be. $Hist_{C_s}(t)$ defines the historical behaviour of group coordinators, including their frequency of participation in consensus, success rate, block generation quality, and contribution to the network. A reduction in score may result from intentional wrongdoing, an inability to maintain consistent online availability, or a failure to comply with consensus agreements.
\par The candidate with the highest score is designated as the validator $V$, who is then responsible for validating the local/global models submitted by all group coordinators.

\subsection{DT-Assisted FL with Blockchain}

\subsubsection{Local Device Classification}
If the inherent capabilities of the LDs are sufficient to meet the needs of local training and model transmission, it is not necessary for the group coordinator to assist in the generation of DTs. Otherwise, DTs will be needed to replace the resource scarce LDs for local training. To achieve this goal, a threshold $threshold (t) $ has been defined to classify the participating LDs in training. The resource score of the local device is calculated based on the formula given in (\ref{eq3}). If the score of $LD_k$ satisfies $Score_{LD_k}(t)\leq threshold(t)$, the group coordinator generates the corresponding DT for $LD_k$, which is then designated as $BLD_j$. The $LD_k$ that does not require the participation of a group coordinator in DT generation is designated as $GLD_i$. This type of local device can independently perform local training of FL and generate jamming signals to assist in the safe synchronization of DT with $BLD_j$.

\begin{equation}\label{eq3}
\begin{aligned}
Score_{LD_k}(t)=&\omega_{Proc}(Cap^{\rm Proc}_{LD_k}(t))+\omega_{Store}(Cap^{\rm Store}_{LD_k}(t))
\\&+\omega_{Com}(Cap^{\rm Com}_{LD_k}(t))+\omega_{Pow}(E_{LD_k}(t)),
\end{aligned}
\end{equation}
where $Cap^{\rm Proc}_{LD_k}(t)$ is the power capability score measured in terms of the CPU power and the CPU cycles per bit required to perform and execute a task. $Cap^{\rm Store}_{LD_k}(t)$ is the storage capacity measured in bits of the $LD_k$, and $Cap^{\rm Com}_{LD_k}(t)$ is the wireless communication distance and bandwidth capabilities of the $LD_k$. $E_{LD_k}(t)$ of a device incorporates the residual battery power, the transmission power of data when communicating with other $LD_k$, and the execution power of the CPU cycle. $\omega_{Proc}$, $\omega_{Store}$, $\omega_{Com}$ and $\omega_{Pow}$ represent the score weight.

\subsubsection{Local Model Training for $GLD_i$ and $BLD_j$}
In the case of $GLD_i$, each $GLD_i$ performs its local model training based on its local data. $D_{i,\rm G}$ represents the size of the $GLD_i$ local sample data in bits, while $v_i$ represents local CPU rate of $GLD_i$. Consequently, the latency required for $GLD_i$ to complete local training can be expressed as
\begin{equation}\label{eq4}
t_i^{\rm loc}=\frac{\eta\iota D_{i,\rm G}}{v_i}, \forall i\in \mathcal{M},
\end{equation}
where $\eta$ and $\iota$ are the number of local training iterations, and the number of CPU cycles for training a single bit, respectively. Accordingly, the energy consumption of $GLD_i$ to complete local training can be expressed as
\begin{equation}\label{eq5}
E_i^{\rm loc}=\kappa \eta \iota D_{i,\rm G}v_i^{2}, \forall i\in \mathcal{M},
\end{equation}
where $\kappa$ is the effective switched capacitance \cite{9860543}. The duration of local model training is equal to $\max_{i\in\mathcal{M}}\{t_i^{\rm loc}\}$, which respresents the duration for all GLDs to complete their respective local training.
\par Similarly, the latency of the DTs generated by the group coordinator for local model training corresponding to $BLD_j$ can be expressed as 

\begin{equation}\label{eq6}
t^{\rm loc}_{\rm B}=\frac{\sum_{j\in \mathcal{N}}\eta\iota D_{j,\rm B}}{v_{\rm c}},
\end{equation}
where $v_{\rm c}$ and $D_{j,\rm B}$ are the CPU rate used by the group coordinator $C_s$ for local training of DTs and the size of $BLD_j$ local training data, respectively.

\subsubsection{Local Model Upload for $GLD_i$ and DT Synchronization with $BLD_j$}
We ultilize non-orthogonal multiple access (NOMA) technique to transmit $GLD_i$ local model to the group coordinator. NOMA enables multiple users to communicate over the same time and frequency resources by grouping users based on channel gain and employing power domain multiplexing. This efficient resource utilization supports a greater number of devices to participate simultaneously in FL, thereby enhancing training efficiency \cite{9036885}. All $\{GLD_i\}_{i\in \mathcal{M}}$ form a NOMA cluster to send their local model to the group coordinator simultaneously. Due to the nature of the uplink NOMA, the decoding order can be artificially determined. For convenience, we assume that the successive interference cancellation order is in reverse order of the $\{GLD_i\}_{i\in \mathcal{M}}$ indices, i.e., $\{I, I-1, \cdots,1\}$. Let $L$ and $t^{\rm up}_{\rm G}$ denote the size of the $\{GLD_i\}_{i\in \mathcal{M}}$ local model data and the model transmission latency, respectively. To complete the transmission of $L$ in $t^{\rm up}_{\rm G}$, the transmission power of $GLD_i$ in the NOMA transmission \cite{9036885} can be expressed as

\begin{equation}\label{eq7}
p_i^{\rm G}=\frac{n_{\rm C}}{g_{i{\rm C}}}(2^{\frac{L}{Wt^{\rm up}_{\rm G}}}-1)2^{\frac{L}{Wt^{\rm up}_{\rm G}}}, \forall i\in \mathcal{M},
\end{equation}
where $n_{\rm C}$ is the background noise at the group coordinator, $W$ is the bandwidth, and $g_{i{\rm C}}$ is the channel power gain from $GLD_i$ to the group coordinator.

\par Therefore, the energy consumption for $GLD_i$ to upload its local model to the group coordinator is given by
\begin{equation}\label{eq8}
E_i^{\rm up}=p_i^{\rm G}t^{\rm up}_{\rm G}, \forall i\in \mathcal{M}.
\end{equation}

\par At the stage where $BLD_j$ uploads update information to the group coordinator, there is a malicious attacker who is eavesdropping or tampering with the data transmission of $BLD_j$. In order to guarantee the security of the transmission, each $GLD_i$ transmits jamming signal with the objective of enhancing the secrecy
throughput of $BLD_j$ for group coordinator transmission. In this work, we assume that all $GLD_i$ are equipped with full-duplex antennas, which facilitate the transmission of jamming signals by each $GLD_i$ while local models are uploaded \cite{9860543}. In particular, the jamming signals $\{q_i\}_{i\in\mathcal{M}}$ are applied to $BLD_j$, and the secrecy rate from $BLD_j$ to the group coordinator can be expressed as

\begin{equation}\label{eq9}
\begin{aligned}
R_j^{\rm sec}=&W[\log_{2}{(1+\frac{p^{\rm B}_jh_{j{\rm C}}}{n_j})}-\\&\log_{2}{(1+\frac{p^{\rm B}_jh_{j{\rm E}}}{n_{\rm E}+\sum_{i\in \mathcal{M}}q_ig_{i{\rm E}}})}]^+, \forall i\in \mathcal{M}, \forall j\in \mathcal{N},
\end{aligned}
\end{equation}
where $g_{i{\rm E}}$ and $n_{\rm E}$ denote the channel power gain from $GLD_i$ to the eavesdropper and the background noise of the eavesdropper, respectively. The notation $[x]^+$ represents $\max\{0, x\}$. 
We use variable $t^{\rm up}_{\rm B}$ to denote the duration of the $\{BLD_j\}_{j \in\mathcal{N}}$ uploading transmission. Correspondingly, each $GLD_i$ energy consumption for its jamming signal can be expressed as
\begin{equation}\label{eq10}
E^{\rm jam}_i=q_it^{\rm up}_{\rm B}, \forall j\in \mathcal{N}.
\end{equation}

\subsubsection{Average Aggregation and Model Upload for Group Coordinators}
Upon receiving local model data from all $GLD_i$ and DTs, the group coordinator aggregates all local models in order to update the global model. In particular, it is assumed that the group coordinator employs a fixed CPU rate and a fixed upload rate. Consequently, the aggregation and upload latency of the group coordinator model is fixed and represented by $T^{\rm agg}_{\rm C}$ and $T^{\rm up}_{\rm C}$. The group coordinator employs the aggregation method FedAVG to derive local/global model \cite{9850408}.

\subsubsection{Blockchain-Based Global Aggregation for Main Coordinator}
During this phase, the group coordinator acts as a participant, with each global model generated by the group coordinator serving as a local model, thus referred to as the local/global model. Validators chosen from the fog devices utilize the shared global model to verify the local/global models uploaded by the group coordinator, ensuring that the final global model used for local training originates from the broadcasts of the main coordinator, and generates blocks for the blockchain. The final global model is obtained by the main coordinator through deep learning techniques. The abundance of computing resources available at the main coordinator, situated in the cloud tier, enables the rapid completion of the final global model calculation. Therefore, we assume that the total time required for generating the final global model and blockchain transaction model data is a negligible positive number $T^{\rm MC}$ in this article \cite{10416386}.
\par The proposed framework of DT-assisted FL with blockchain is summarized in Algorithm \ref{alg:FL}.

\begin{algorithm}[!h]
    \caption{DT-Assisted FL with Blockchain}
    \label{alg:FL}
    \renewcommand{\algorithmicrequire}{\textbf{Input:}}
    \renewcommand{\algorithmicensure}{\textbf{Output:}}
    
    \begin{algorithmic}[1]
        \REQUIRE The training process involves the participation of local devices $\{LD_1, LD_2,\cdots, LD_k\}$ in each cluster and the group coordinator $\{C_s\}_{s\in \mathcal{S}}$.

        \ENSURE Final global model.   
    \FOR{each iteration $t=0, 1, 2,\cdots$}
        \FOR{each participating group coordinator $\{C_s\}_{s\in \mathcal{S}}$}
            \STATE Main Coordinator computes reputation score $Score^{\rm Rep}_{C_s}$ according to (\ref{eq2}).
            \STATE Main coordinator selects highest scoring group coordinator as validator $V$.
        \ENDFOR
        \FOR{each cluster $s=0, 1, 2,\cdots, \mathcal{S}$}
            \STATE According to (\ref{eq4}), classify local devices.
            \STATE The DT of $BLD_j$ generated by the group coordinator according to (\ref{eq1}).
            \STATE $\{GLD_i\}_{i\in \mathcal{M}}$ conducts local training and subsequently uploads the local model to the group coordinator. Concurrently, $\{GLD_i\}_{i\in \mathcal{M}}$ transmits cooperative jamming signals to $\{BLD_j\}_{j\in \mathcal{N}}$, thereby guaranteeing the secure upload of DTs synchronisation data by $\{BLD_j\}_{j\in \mathcal{N}}$. Subsequently, DTs initiate local training.
            \STATE The group coordinator executes average aggregation through FedAVG to obtain the local/global model and broadcasts it to validator $V$.
        \ENDFOR
        \STATE Verify the local/global model with validator $V$ and create new block for the transactions.
        \STATE Add the block to the blockchain.
        \STATE Perform deep learning on verified models with main coordinator and update the final global model.
    \ENDFOR
    \RETURN Final global model
    \end{algorithmic}
\end{algorithm}
\subsection{Problem Formulation}
Based on our modeling in the previous section, the total delay of the validator receiving local/global models in each round of FL iteration can be expressed as follows:

\begin{equation}\label{eq11}
T=\mathop{\max}\{T^{\rm GLD}, T^{\rm BLD}\}+T^{\rm agg}_{\rm C}+T^{\rm up}_{\rm C}+T^{\rm MC},
\end{equation}
\begin{equation}\label{eq12}
T^{\rm GLD}=\mathop{\max}_{i\in \mathcal{M}}\{t_i^{\rm loc}\}+t^{\rm up}_{\rm G},
\end{equation}
\begin{equation}\label{eq13}
T^{\rm BLD}=t^{\rm up}_{\rm B}+t^{\rm loc}_{\rm B}.
\end{equation}

The objective is to accelerate the FL process of the group coordinator by minimising the total delay $T$. This is to be achieved by jointly optimising the local computing rate $v_i$ and the jamming transmit powers $q_i$ of $GLD_i$, as well as the upload time $t^{\rm up}_{\rm G}$ of local model data and the DT update data time $t^{\rm up}_{\rm B}$ corresponding to $\{BLD_j\}_{j\in \mathcal{N}}$. In light of the aforementioned considerations, we have established the following delay, which can be expressed as

\begin{subequations}\label{eq14}
\begin{align}
\text{\textbf{P}:} \min T  \ \ \tag{\ref{eq14}}
\end{align} 
\begin{alignat}{2}
\text{s.t. } & E_i^{\rm loc}+E_i^{\rm up}+E_i^{\rm jam}\leq E_i^{\rm max}, \forall i\in \mathcal{M},\label{eq14a}   \\
     &0\leq v_i \leq V_i^{\rm max}, \forall i\in \mathcal{M},\label{eq14b}\\
     &0\leq p_i^{\rm G} \leq P_i^{\rm max}, \forall i\in \mathcal{M},\label{eq14c}\\
     &0\leq q_i \leq Q_i^{\rm max}, \forall i\in \mathcal{M}.\label{eq14d}
\end{alignat}
\qquad\; \text{variables:} $v_i$, $q_i$, $t^{\rm up}_{\rm G}>0$, $t^{\rm up}_{\rm B}>0$, $\forall i\in \mathcal{M}$.
\end{subequations}
\par Constraint (\ref{eq14a}) ensures that each $GLD_i$’s total energy consumption cannot exceed its energy capacity $E_i^{\rm max}$. Constraint (\ref{eq14b}) means that each $GLD_i$’s CPU rate cannot exceed its maximum rate $V_i^{\rm max}$. Constraint (\ref{eq14c}) means that each $GLD_i$’s uploading power $p_i^{\rm G}$ cannot exceed its maximum transmit power $P_i^{\rm max}$. Constraint (\ref{eq14d}) means that the energy consumption of the jamming power cannot exceed its energy capacity $Q_i^{\rm max}$.

\section{Delay Minimized FL Schemes}
Since the objective function is non-convex, the problem \textbf{P} is a non-convex optimization problem. Consequently, in order to facilitate the solution, we decompose problem \textbf{P}.
\subsection{Auxiliary Variables and Equivalent Transformations}
Let $y$ be the upper bound of all $GLD_i$s’ local training latency $\{t_i^{\rm loc}\}_{i\in\mathcal{M}}$, i.e.,

\begin{equation}\label{eq15}
     \mathop{\max}_{i\in\mathcal{M}} \{t_i^{\rm loc}\}\leq y.
\end{equation}

\par Based on (\ref{eq4}) and (\ref{eq15}), we can derive the lower bound of $GLD_i$’s local processing rate as
\begin{equation}\label{eq16}
     \frac{\eta\iota D_{i,\rm G}}{y}\leq v_i, \forall i\in \mathcal{M}.
\end{equation}

\par According to (\ref{eq5}), $GLD_i$’s energy consumption for its local model training increases with the increase of $v_i$. Thus, to minimize the energy consumption of $GLD_i$ for its local model training and meet the constraints, $v_i$ needs to satisfy the following condition
\begin{equation}\label{eq17}
     v_i=\frac{\eta\iota D_{i,\rm G}}{y}\leq V_i^{\rm max}, \forall i\in \mathcal{M}.
\end{equation}

\par Therefore, (\ref{eq16}) is strictly active while satisfying constraint (\ref{eq14b}) before.
Moreover, $LD_i$’s energy consumption for its local model training can be expressed as
\begin{equation}\label{eq18}
     E_i^{\rm loc}=\frac{\kappa\eta^3\iota^3 D_{i,\rm G}^3}{y^2}, \forall i\in \mathcal{M}.
\end{equation}

\par Note that constraint (\ref{eq16}) leads to an equivalent lower bound on $y$, i.e.,
\begin{equation}\label{eq19}
     \mathop{\max}_{i\in\mathcal{M}} \{\frac{\eta\iota D_{i,\rm G}}{V_i^{\rm max}}\}\leq y.
\end{equation}

\par Furthermore, we introduce another variable $Q$, to denote the aggregate jamming effect perceived by the eavesdropper. This is written as
\begin{equation}\label{eq20}
Q=\sum_{i\in \mathcal{M}}q_ig_{i{\rm E}}.
\end{equation}
\par By substituting of $Q$ into (\ref{eq9}), we can derive the secure rate of the $BLD_j$ to group coordinator as

\begin{equation}\label{eq21}
\begin{aligned}
R_j^{\rm sec}=&W[\log_{2}{(1+\frac{p^{\rm B}_jh_{j{\rm C}}}{n_j})}-\\&\log_{2}{(1+\frac{p^{\rm B}_jh_{j{\rm E}}}{n_{\rm E}+Q})}]^+, \forall j\in \mathcal{N}.
\end{aligned}
\end{equation}
\par Notice that to ensure a non-zero $R_j^{\rm sec}$ for each $BLD_j$, $Q$ should satisfy the following constraint
\begin{equation}\label{eq22}
\mathop{\max}_{j\in\mathcal{N}} \{\frac{n_jh_{j{\rm E}}}{h_{j{\rm C}}}\}-n_{\rm E}<Q.
\end{equation}
\par By using (\ref{eq21}) and (\ref{eq22}), we can derive the lower bound of $t^{\rm up}_{\rm B}$ under the given $Q$ as

\begin{equation}\label{eq23}
\begin{aligned}
\widehat{t}^{\rm up}_{\rm B}&=\mathop{\max}_{j\in\mathcal{N}}\{\frac{D_{j,\rm B}}{R^{\rm sec}_j}\}\\&=\mathop{\max}_{j\in\mathcal{N}}\{\frac{D_{j,\rm B}}{W[\log_{2}{(1+\frac{p^{\rm B}_jh_{j{\rm C}}}{n_j})}-\log_{2}{(1+\frac{p^{\rm B}_jh_{j{\rm E}}}{n_{\rm E}+Q})}]}\}.
\end{aligned}
\end{equation}

\subsection{Problem Transformation}

In problem \textbf{P}, there is variable coupling between $T^{\rm GLD}$ and $T^{\rm BLD}$, with $T^{\rm BLD}$ being constrained by variable $Q$. To this end, we decompose problem \textbf{P}, fix $T^{\rm BLD}$ first, solve for the optimal solution of $T^{\rm GLD}$, and then perform a linear search to obtain the optimal $\min\max\{T^{\rm GLD}, T^{\rm BLD}\}$. The problem \textbf{P} can be transformed as follows:

\begin{subequations}\label{eq24}
\begin{align}
\text{\textbf{P}-GLD: } \widehat{T}^{\rm GLD}=\min y+t^{\rm up}_{\rm G}+T^{\rm agg}_{\rm C}+T^{\rm up}_{\rm C}  \ \ \tag{\ref{eq24}}
\end{align} 
\begin{alignat}{2}
\text{s.t. } & \frac{\kappa\eta^3\iota^3 D_{i,\rm G}^3}{y^2}+\frac{n_{\rm C}}{g_{i{\rm C}}}(2^{\frac{L}{Wt^{\rm up}_{\rm G}}}-1)2^{\frac{(i-1)L}{Wt^{\rm up}_{\rm G}}}t^{\rm up}_{\rm G}\notag
\\ &+\widehat{t}^{\rm up}_{\rm B}q_i\leq E_i^{\rm max}, \forall i\in \mathcal{M},\label{eq24a}   \\
     &Q=\sum_{i\in \mathcal{M}}q_ig_{i{\rm E}},\label{eq24b}\\
     &\underline{t}^{\rm up}_{\rm G} \leq t^{\rm up}_{\rm G}=\{{\rm argmin}\;t^{\rm up}_{\rm G}|\frac{n_{\rm C}}{g_{i{\rm C}}}(2^{\frac{L}{Wt^{\rm up}_{\rm G}}}-1) \notag
     \\ & \times2^{\frac{L}{Wt^{\rm up}_{\rm G}}}\leq P_i^{\rm max}, \forall i\in \mathcal{M}\},\label{eq24c}\\
     &0<q_i \leq Q_i^{\rm max}, \forall i\in \mathcal{M},\label{eq24d}
\end{alignat}
\qquad\; \text{variables:} $y$, $q_i$, $t^{\rm up}_{\rm G}>0$, $\forall i\in \mathcal{M}$.
\end{subequations}

\par In problem \textbf{P}-GLD, constraint (\ref{eq24a}) is derived from (\ref{eq14a}) before by exploiting $\widehat{t}^{\rm up}_{\rm B}$ in (\ref{eq23}) under the given $Q$. Constraint (\ref{eq24b}) comes from the definition of $Q$ in (\ref{eq20}). Constraint (\ref{eq24c}) comes from constraint (\ref{eq14a}). According to (\ref{eq7}), exploiting the decreasing feature of $p_i^{\rm G}$ with respect to $t^{\rm up}_{\rm G}$, constraint (\ref{eq14c}) leads to an equivalent lower-bound on $t^{\rm up}_{\rm G}$ as shown in constraint (\ref{eq24c}).

\begin{proposition}
Problem \textbf{P}-GLD is a convex optimization problem with respect to $y$, $t^{\rm up}_{\rm G}$ and $\{q_i\}_{i\in \mathcal{M}}$.
\end{proposition}

\begin{proof}
Please refer to Appendix A.$\hfill\blacksquare$ 
\end{proof}

\subsection{Optimized Solution for Transformed Problem \textbf{P}-GLD}
Based on convexity of \textbf{P}-GLD, we adopt the Karush-Kuhn-Tucker (KKT) conditions to obtain the optimal solutions \cite{5447064}. An important feature of problem \textbf{P}-GLD is that only constraint (\ref{eq24a}) couples all the optimization variables. Thus, we introduce a set of dual variables $\boldsymbol{\lambda}=\{\lambda_i\}_{i\in\mathcal{M}}$ to relax constraint (\ref{eq24a}). With $\boldsymbol{\lambda}$, we can express the Lagrangian function as follows:

\begin{equation}\label{eq25}
\begin{aligned}
     \mathcal{L}(y, t^{\rm up}_{\rm G}, \{q_i\}_{i\in\mathcal{M}}, \boldsymbol{\lambda})&=y+t^{\rm up}_{\rm G}+
\sum_{i\in\mathcal{M}}\lambda_i[\frac{\kappa\eta^3\iota^3 D_{i,\rm G}^3}{y^2}\\&+\frac{n_{\rm C}}{g_{i{\rm C}}}(2^{\frac{L}{Wt^{\rm up}_{\rm G}}}-1)2^{\frac{L}{Wt^{\rm up}_{\rm G}}}t^{\rm up}_{\rm G}
\\ &+\widehat{t}^{\rm up}_{\rm B}q_i-E_i^{\rm max}].
\end{aligned}
\end{equation}

\par With (\ref{eq25}) and given $\boldsymbol{\lambda}$, the primal problem of problem \textbf{P}-GLD can be expressed as

\begin{subequations}\label{eq26}
\begin{align}
\text{(\textbf{P}-GLD-Primal): } \min\mathcal{L}(y, t^{\rm up}_{\rm G}, \{q_i\}_{i\in\mathcal{M}}, \boldsymbol{\lambda}) \ \ \tag{\ref{eq26}}
\end{align} 
\begin{alignat}{2}
\text{s.t. } \text{Constraint (\ref{eq24b})-(\ref{eq24d}) and (\ref{eq19}).} \notag
\end{alignat}
\end{subequations}
\par As the variables in problem \textbf{P}-GLD-Primal exhibit a linear structure, we can decompose this problem into three sub-problems that separately optimize $y$, $t^{\rm up}_{\rm G}$ and $\{q_i\}_{i\in\mathcal{M}}$ as follows:

\subsubsection{Sub-Problem to Optimize $y$ (Sub-Y)}
Given $\{\lambda_i\}_{i\in\mathcal{M}}$, we optimize $y$ according to the following single-variable optimization problem.
\begin{subequations}\label{eq27}
\begin{align}
\text{(Sub-Y): } \min f_Y(y)=y+\sum_{i\in\mathcal{M}}\frac{\lambda_i\kappa\eta^3\iota^3 D_{i,\rm G}^3}{y^2} \ \ \tag{\ref{eq27}}
\end{align} 
\begin{alignat}{2}
\text{s.t. } \text{Constraint (\ref{eq19}).} \notag
\end{alignat}
\end{subequations}
\par In the context of Sub-Y, the first-order derivative of the objective function with respect to $y$ is represented by $\frac{\partial f_Y(y)}{\partial {y}}=1-2\sum_{i\in\mathcal{N}}\frac{\lambda_i\kappa\eta^3\iota^3 D_{i,\rm G}^3}{y^3}$. The optimized value for $y$ can be obtained, which can be expressed as follows:

\begin{equation}\label{eq28}
    y=(2\sum_{i\in\mathcal{N}}\lambda_i\kappa\eta^3\iota^3 D_{i,\rm G}^3)^\frac{1}{3}.
\end{equation}
\par Considering constraint (\ref{eq19}), the optimal solution is given by
\begin{equation}\label{eq29}
    y^*=\max \{(2\sum_{i\in\mathcal{N}}\lambda_i\kappa\eta^3\iota^3 D_{i,\rm G}^3)^\frac{1}{3}, \mathop{\max}_{i\in\mathcal{N}} \{\frac{\eta\iota D_{i,\rm G}}{V_i^{\rm max}}\}\}.
\end{equation}

\subsubsection{Sub-Problem to Optimize $t^{\rm up}_{\rm G}$ (Sub-T)}
Given $\{\lambda_i\}_{i\in\mathcal{M}}$, we optimize variable $t^{\rm up}_{\rm G}$ according to the single-variable optimization problem expressed as follows:
\begin{subequations}\label{eq30}
\begin{align}
\text{(Sub-T): } &\min f_T(t^{\rm up}_{\rm G})=t^{\rm up}_{\rm G}+\notag\\&\sum_{i\in\mathcal{M}}\frac{\lambda_in_{\rm C}}{g_{i{\rm C}}}(2^{\frac{L}{Wt^{\rm up}_{\rm G}}}-1)2^{\frac{(i-1)L}{Wt^{\rm up}_{\rm G}}}t^{\rm up}_{\rm G} \ \ \tag{\ref{eq30}}
\end{align} 
\begin{alignat}{2}
\text{s.t. } \text{Constraint (\ref{eq24c}).} \notag
\end{alignat}
\end{subequations}

\par With (\ref{eq38}), we can derive that the objective function $f_T(t^{\rm up}_{\rm G})$ is convex, and the first order derivative of the objective function with respect to $t^{\rm up}_{\rm G}$ is represented as follows:
\begin{equation}\label{eq31}
\begin{aligned}
    \frac{\partial f_T(t^{\rm up}_{\rm G})}{\partial {t^{\rm up}_{\rm G}}}&=1+\sum_{i\in\mathcal{M}}\frac{\lambda_in_{\rm C}2^{\frac{(i-1)L}{Wt^{\rm up}_{\rm G}}}}{g_{i{\rm C}}}[(2^{\frac{L}{Wt^{\rm up}_{\rm G}}}-1)\\&-(2^{\frac{L}{Wt^{\rm up}_{\rm G}}}\frac{L}{Wt^{\rm up}_{\rm G}}+(2^{\frac{L}{Wt^{\rm up}_{\rm G}}}-1)\frac{(i-1)L}{Wt^{\rm up}_{\rm G}})\ln{2}].
\end{aligned}
\end{equation}
\par Different from Sub-Y, we cannot analytically derive the solution for $\frac{\partial f_T(t^{\rm up}_{\rm G})}{\partial {t^{\rm up}_{\rm G}}}=0$ in (\ref{eq31}). However, by exploiting the convexity of $f_T(t^{\rm up}_{\rm G})$, we can use the bisection search \cite{10129092} to find $\widehat{t}^{\rm up}_{\rm G}=\{t^{\rm up}_{\rm G}|\frac{\partial f_T(t^{\rm up}_{\rm G})}{\partial {t^{\rm up}_{\rm G}}}=0\}$. Combined with constraint (\ref{eq24c}), we obtain the optimal solution of $t^{\rm up}_{\rm G}$ as follows:

\begin{equation}\label{eq32}
    t^{\rm up*}_{\rm G}=\max \{\underline{t}^{\rm up}_{\rm G}, \widehat{t}^{\rm up}_{\rm G}\}.
\end{equation}

\subsubsection{Sub-Problem to Optimize $\{q_i\}_{i\in \mathcal{M}}$ (Sub-Q)}
Given $Q$, Sub-Q can be formulated as a linear programming problem. Consequently, the Lagrange multiplier method is employed to solve $\{q_i\}_{i\in \mathcal{M}}$. A dual variable $\mu$ is introduced for constraint (\ref{eq24b}). Therefore, Sub-Q can be relaxed as follows:
\begin{subequations}\label{eq33}
\begin{align}
\text{(Sub-Q): } \sum_{i\in\mathcal{M}}\lambda_i\widehat{t}^{\rm up}_{\rm B}q_i-\mu(Q-\sum_{i\in \mathcal{M}}q_ig_{i{\rm E}}) \ \ \tag{\ref{eq33}}
\end{align} 
\begin{alignat}{2}
\text{s.t. } \text{Constraint (\ref{eq24d}).} \notag
\end{alignat}
\end{subequations}
\par Due to $\sum_{i\in\mathcal{M}}\lambda_i\widehat{t}^{\rm up}_{\rm B}q_i-\mu(Q-\sum_{i\in \mathcal{M}}q_ig_{i{\rm E}})$ can be rearranged as 
$\sum_{i\in\mathcal{M}}(\lambda_i\widehat{t}^{\rm up}_{\rm B}-\mu g_{i{\rm E}})q_i+\mu Q$, by first setting the value of $\mu$, then the corresponding solution for $\{q_i\}_{i\in \mathcal{M}}$ can be obtained and is expressed as follows:

\begin{equation}\label{eq34}
q_i=
\begin{cases}
0, & \lambda_i\widehat{t}^{\rm up}_{\rm B}-\mu g_{i{\rm E}}>0 \\
{\rm any\ value}\in [0, Q^{\rm max}_i], & \lambda_i\widehat{t}^{\rm up}_{\rm B}-\mu g_{i{\rm E}}=0\\
Q_i^{\rm max}, & \lambda_i\widehat{t}^{\rm up}_{\rm B}-\mu g_{i{\rm E}}<0
\end{cases}.
\end{equation}

\par In accordance with the results of (\ref{eq34}), the optimal value of $\mu$ can be identified within the set $\{\frac{\lambda_i\widehat{t}^{\rm up}_{\rm B}}{g_{i{\rm E}}}\}_{i\in \mathcal{M}}$. Consequently, the value of $q_i$ (assuming $q_i=\frac{\lambda_i\widehat{t}^{\rm up}_{\rm B}}{g_{i{\rm E}}}\}$) can be adjusted to satisfy the constraints set forth in (\ref{eq24b}). In order to achieve a more accurate representation, the values of $\{\frac{\lambda_i\widehat{t}^{\rm up}_{\rm B}}{g_{i{\rm E}}}\}_{i\in \mathcal{M}}$ are sorted in descending order, and a mapping $map(i)$ is introduced to represent the index of the sorted $GLD_i$. In this context, the value of $\{\frac{\lambda_i\widehat{t}^{\rm up}_{\rm B}}{g_{i{\rm E}}}\}_{i\in \mathcal{M}}$ is defined as the $map(i)$-th largest value after reordering. Therefore, we can obtain the following proposition.

\begin{proposition}
There is a unique value of $\widehat{i}$, and that the optimal solution of Sub-Q can be expressed as follows:

\begin{equation}\label{eq35}
q_i=
\begin{cases}
0, & map(i)\leq map(\widehat{i})-1 \\
\frac{Q-\sum_{map(r)\ge map(\widehat{i}+1)}q_rg_{r{\rm E}}}{g_{\widehat{i}{\rm E}}}, & i=\widehat{i}\\
Q_i^{\rm max}, & map(i)\ge map(\widehat{i})+1
\end{cases}.
\end{equation}

\end{proposition}

\begin{proof}
Please refer to Appendix B. 
\end{proof}

\par In accordance with Proposition 2, the solution to Sub-Q can be obtained by enumerating the values of $map(i)$ in accordance with (\ref{eq35}). In particular, the value of $q_i=Q_i^{\rm max}$, $map(i)\ge map(\widehat{i})+1$, $\forall i\in \mathcal{M}$ is set until $\frac{Q-\sum_{map(r)\ge map(\widehat{i})+1}q_rg_{r{\rm E}}}{g_{\widehat{i}{\rm E}}}\leq Q_{\widehat{i}}^{\rm max}$ is reached. Subsequently, the value of $\widehat{i}$ is obtained, after which the value of $\{q_i\}_{i\in \mathcal{M}}$ is set in accordance with (\ref{eq35}). It should be noted that, in accordance with $T^{\rm GLD}>T^{\rm BLD}$, constraint (\ref{eq22}) and constraint (\ref{eq24d}), Sub-Q is always feasible.

\subsubsection{Dual variable $\boldsymbol{\lambda}$ update}
\par With the optimal solutions of $y^*$, $t^{\rm up*}_G$ and $\{q_i^*\}_{i\in \mathcal{M}}$, we can express the dual problem as $\max\mathcal{L}(y^*, t^{\rm up*}_G, \{q_i^*\}_{i\in \mathcal{M}}, \boldsymbol{\lambda})$. Then, we adopt the sub-gradient method \cite{4749425} for solving the dual problem, namely, we update $\lambda_i$ as follows:
\begin{equation}\label{eq36}
\begin{aligned}
\lambda_i&=[\lambda_i+\alpha_n(\frac{\kappa\eta^3\iota^3 D_{i,\rm G}^3}{(y^*)^2}\\&+\frac{n_{\rm C}}{g_{i{\rm C}}}(2^{\frac{L}{Wt^{\rm up*}_{\rm G}}}-1)2^{\frac{L}{Wt^{\rm up*}_{\rm G}}}t^{\rm up*}_{\rm G}
\\ &+\widehat{t}^{\rm up}_{\rm B}q^*_i-E_i^{\rm max})]^+, \forall i\in \mathcal{M}.
\end{aligned}
\end{equation}
\par $\alpha_n$ denotes the step size for updating the dual variables, and we adopt a decreasing step-size defined as $\alpha_n=\frac{\alpha_0}{\sqrt{n}}$ (where $\alpha_0$ is constant and $n$ is the iteration index). For convex optimization problems, this step-size can ensure convergence to the optimal solution. Such an iterative procedure continues until it converges in the variables $\{\lambda_i\}_{i\in\mathcal{M}}$.

\subsection{Optimized Solution for Original Problem \textbf{P}}

\par For each given $Q$, the optimal values of $y^*$, $t^{\rm up}_{\rm G}$, and $\{q_i^*\}_{i\in \mathcal{M}}$, as well as the corresponding values of $\widehat{T}^{\rm GLD}$, can be obtained by solving \textbf{P}-GLD. Nevertheless, the coupling effect between $T^{\rm GLD}$ and $T^{\rm BLD}$ is variable, as it is not possible to express their objective functions $\min\max\{T^{\rm GLD}, T^{\rm BLD}\}$ by analytical results. By leveraging the fundamental principles of univariate optimisation, a linear search can be conducted on $Q$ within the specified feasible interval, defined by the constraint $\mathop{\max}_{j\in\mathcal{N}} \{\frac{n_jh_{j{\rm E}}}{h_{j{\rm C}}}\}-n_{\rm E}<Q\leq \sum_{i \in \mathcal{M}}Q_i^{\rm max}g_{i{\rm E}}$. For each value of $Q$, the value of $\widehat{T}^{GLD}$ and $T^{\rm BLD}$ are calculated using (\ref{eq29}), (\ref{eq32}), (\ref{eq35}) and (\ref{eq13}). Ultimately, based on the calculated $\widehat{T}^{GLD}$ and $T^{\rm BLD}$ in (\ref{eq11}), $T$ should be calculated and the optimal $T^*$ obtained through an iterative process. Algorithm \ref{alg:P} offers a concise overview of the solution process for problem \textbf{P}.

\begin{algorithm}[!h]
    \caption{Proposed Algorithm for Problem (\textbf{P})}
    \label{alg:P}
    \renewcommand{\algorithmicrequire}{\textbf{Input:}}
    \renewcommand{\algorithmicensure}{\textbf{Output:}}
    
    \begin{algorithmic}[1]
        \REQUIRE Set $T^*=\infty$, $\theta_1$ as the step size of the linear search, $\theta_2$ as a small positive number and initializes the values $\{\lambda_i\}_{i\in \mathcal{M}}$.

        \ENSURE $T^*$, $t_{\rm B}^{\rm up, FL}$, $y^{\rm FL}$, $t_{\rm G}^{\rm up, FL}$ and $\{q_i^{\rm FL}\}_{i\in \mathcal{M}}$.   
    \FOR{$Q=\mathop{\max}_{j\in\mathcal{N}} \{\frac{n_jh_{j{\rm E}}}{h_{j{\rm C}}}\}-n_{\rm E}:\theta_1:\sum_{i \in \mathcal{M}}Q_i^{\rm max}g_{i{\rm E}}$} 
        \REPEAT
        \STATE  Given $\{\lambda_i\}_{i\in\mathcal{N}}$, use (\ref{eq29}) and (\ref{eq32}) to compute $y^*$ and $t^{\rm up*}_{\rm G}$, respectively. 
        \STATE Given $\{\lambda_i\}_{i\in\mathcal{N}}$ and set the value of $q_i=Q_i^{\rm max}$, $map(i)\ge map(\widehat{i})+1$, $\forall i\in \mathcal{M}$ until the condition $\frac{Q-\sum_{map(r)\ge map(\widehat{i})+1}q_rg_{r{\rm E}}}{g_{\widehat{i}{\rm E}}}\leq Q_{\widehat{i}}^{\rm max}$ is met. Obtain the value of $\widehat{i}$, after which the value of $\{q_i^*\}_{i\in \mathcal{M}}$ is obtained from (\ref{eq35}).
        \STATE With $y^*$, $t^{\rm up*}_{\rm G}$ and $\{q_i^*\}_{i\in \mathcal{M}}$ update $\{\lambda_i\}_{i\in\mathcal{N}}$ according to (\ref{eq36}).
        \UNTIL $\mathop{\max}_{i\in\mathcal{M}} |\lambda_i-\lambda_i^*|<\theta_2$.
        \STATE Computes $\widehat{T}^{\rm GLD}$ and $T^{\rm BLD}$ according to (\ref{eq24}).
        \STATE Computes $T$ according to (\ref{eq11}).
        \IF{$T^*>T$}
            \STATE Set $T^*=T$, $t_{\rm B}^{\rm up, FL}=\widehat{t}_{\rm B}^{\rm up}$, $y^{\rm FL}=y^*$, $t_{\rm G}^{\rm up, FL}=t_{\rm G}^{\rm up*}$ and $\{q_i^{\rm FL}\}_{i\in \mathcal{M}}=\{q_i^*\}_{i\in \mathcal{M}}$.
        \ENDIF

    \ENDFOR
    \RETURN $T^*$, $t_{\rm B}^{\rm up, FL}$, $y^{\rm FL}$, $t_{\rm G}^{\rm up, FL}$ and $\{q_i^{\rm FL}\}_{i\in \mathcal{M}}$
    \end{algorithmic}
\end{algorithm}

\begin{table*}[!htbp]
\begin{center}
\caption{Parameter Configurations and Simulator Settings}
\renewcommand{\arraystretch}{1.25}
\begin{tabular}{c|l}
	\toprule
	\multicolumn{1}{p{3cm}}{\centering Parameters}
	&\multicolumn{1}{p{2.5cm}}{\centering Numerical values}\\
	\midrule
    Participating local devices & 6 GLDs and 4 BLDs\\
    $\{D_{1,\rm G}, D_{2,\rm G}, D_{3,\rm G}, D_{4,\rm G}, D_{5,\rm G}, D_{6,\rm G}\}$& $\{30, 45, 40, 50, 55, 35\}$ Mbits \\
    $\{g_{1{\rm C}}, g_{2{\rm C}}, g_{3{\rm C}}, g_{4{\rm C}}, g_{5{\rm C}}, g_{6{\rm C}}\}$& $\{2.3, 2.5, 2.4, 2.2, 2.7, 2.6\}\times 10^{-8}$  \\
    $\{g_{1{\rm E}}, g_{2{\rm E}}, g_{3{\rm E}}, g_{4{\rm E}}, g_{5{\rm E}}, g_{6{\rm E}}\}$& $\{1.6, 1.3, 1.7, 1.4, 1.2, 1.5\}\times 10^{-8}$ \\
    $\{D_{1,\rm B},D_{2,\rm B}, D_{3,\rm B}, D_{4,\rm B}\}$& $\{2.0, 3.5, 3.0, 2.5\}$ Mbits \\
    $\{h_{1{\rm C}}, h_{2{\rm C}}, h_{3{\rm C}}, h_{4{\rm C}}\}$& $\{1.0, 0.8, 1.1, 0.9\}\times 10^{-8}$\\
    $\{h_{1{\rm E}}, h_{2{\rm E}}, h_{3{\rm E}}, h_{4{\rm E}}\}$& $\{0.95, 0.85, 1.15, 1.05\}\times 10^{-9}$ \\
    $\{P_1^{\rm max}, P_2^{\rm max}, P_3^{\rm max}, P_4^{\rm max}, P_5^{\rm max}, P_6^{\rm max}\}$& $\{1.9, 2.1, 1.8, 2.0, 2.2, 1.7\}$ W \\
    $\{Q_1^{\rm max}, Q_2^{\rm max}, Q_3^{\rm max}, Q_4^{\rm max}, Q_5^{\rm max}, Q_6^{\rm max}\}$ & $\{1.1, 0.8, 0.9, 1.2, 0.7, 1.0\}$ W \\
    $\{V_1^{\rm max}, V_2^{\rm max}, V_3^{\rm max}, V_4^{\rm max}, V_5^{\rm max}, V_6^{\rm max}\}$& $\{1.0, 1.2, 1.4, 1.6, 1.8, 2.0, 2.2\}\times 10^{9}$ Hz\\
    $\{E_1^{\rm max}, E_2^{\rm max}, E_3^{\rm max}, E_4^{\rm max}, E_5^{\rm max}, E_6^{\rm max}\}$ & $\{3.8, 4.0, 3.4, 3.6, 3.8, 3.2\}$ J\\
    $\{p^{\rm B}_1, p^{\rm B}_2, p^{\rm B}_3, p^{\rm B}_4\}$ & $\{1.6, 1.4, 1.3, 1.5\}$ W \\
    Dataset & MNIST\\
    Training model& CNN\\
    \bottomrule
\end{tabular}\label{tb:notation}
\end{center}
\end{table*}

\section{Performance Evaluation}
This section presents the simulation results of the DT and blockchain-assisted FL scheme. In order to evaluate the performance of the proposed FL scheme, we consider a multi-tier computing network distributed in a simulated area of 1000m$\times$1000m, which consists of a total of 10 industrial local device clusters, each associated with a group coordinator, and a main coordinator, acting as a server. Hyperledger Fabric 2.3.3 network testing was conducted using Alibaba Cloud's virtual machine in order to evaluate the blockchain performance of the proposed FL solution \cite{9916259}. Table \ref{tb:notation} provides a summary of the parameter configuration and simulation settings. All the results are obtained with a PC of Intel(R) Xeon(R) CPU E5-2670v2 @2.50GHz.
\begin{figure}[htbp]
\centering
\subfigure[$Q=3\times10^{-8}$.]{\includegraphics[width=6.5cm]{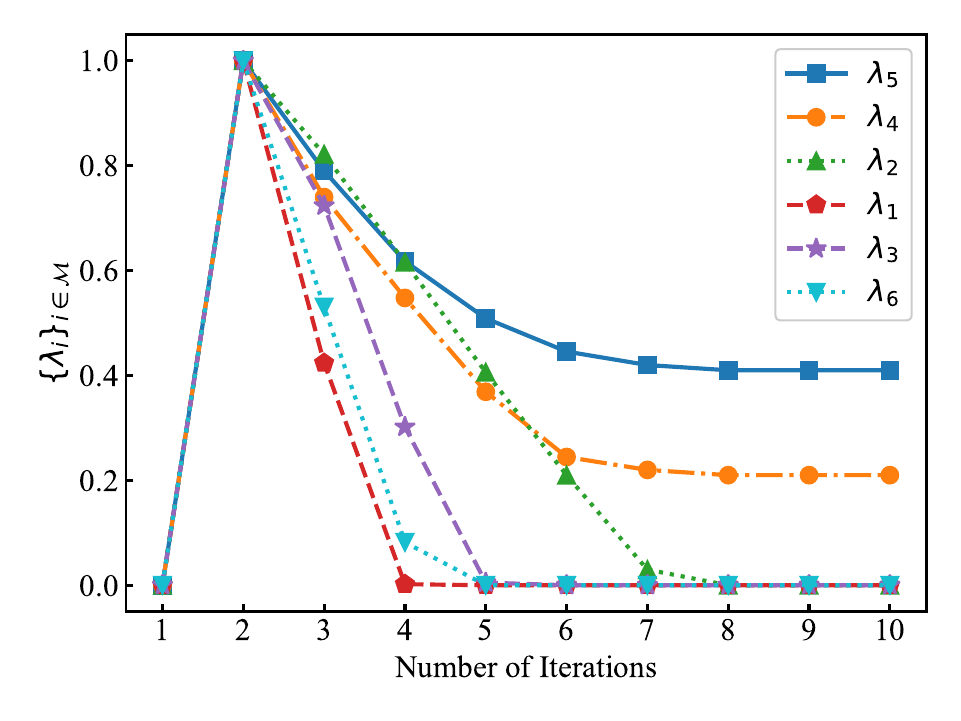}}
\subfigure[$Q=3\times10^{-7}$.]{\includegraphics[width=6.5cm]{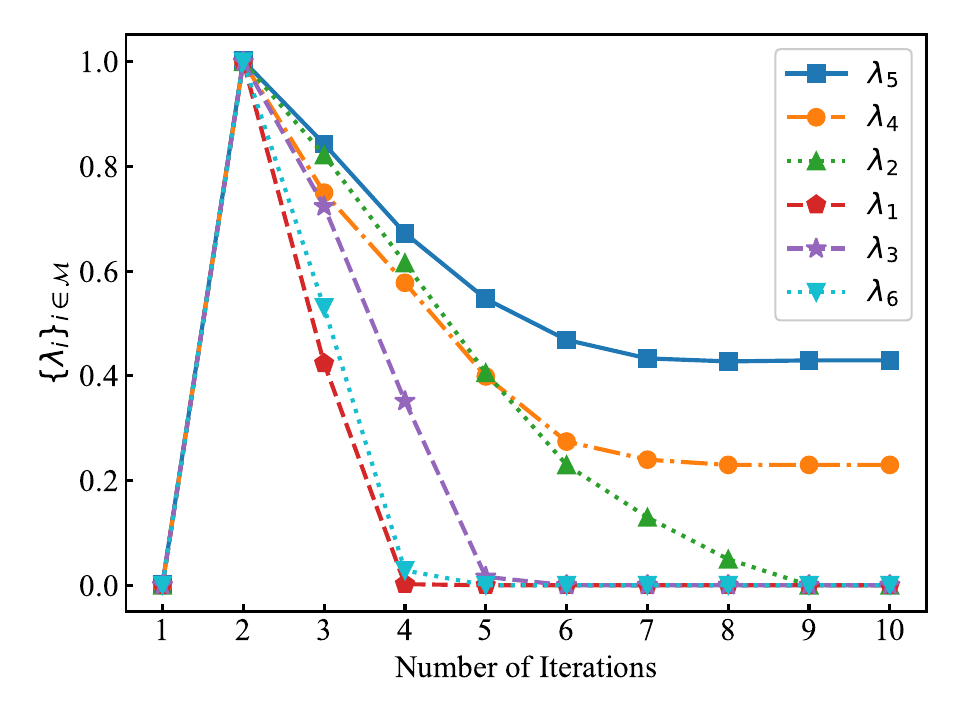}}
\caption{Convergence of $\{\lambda_i\}_{i\in \mathcal{M}}$ in Algorithm \ref{alg:P}.}
\label{E1}
\end{figure}

\begin{figure}[htbp]
    \centering
    \includegraphics[width=6.1cm]{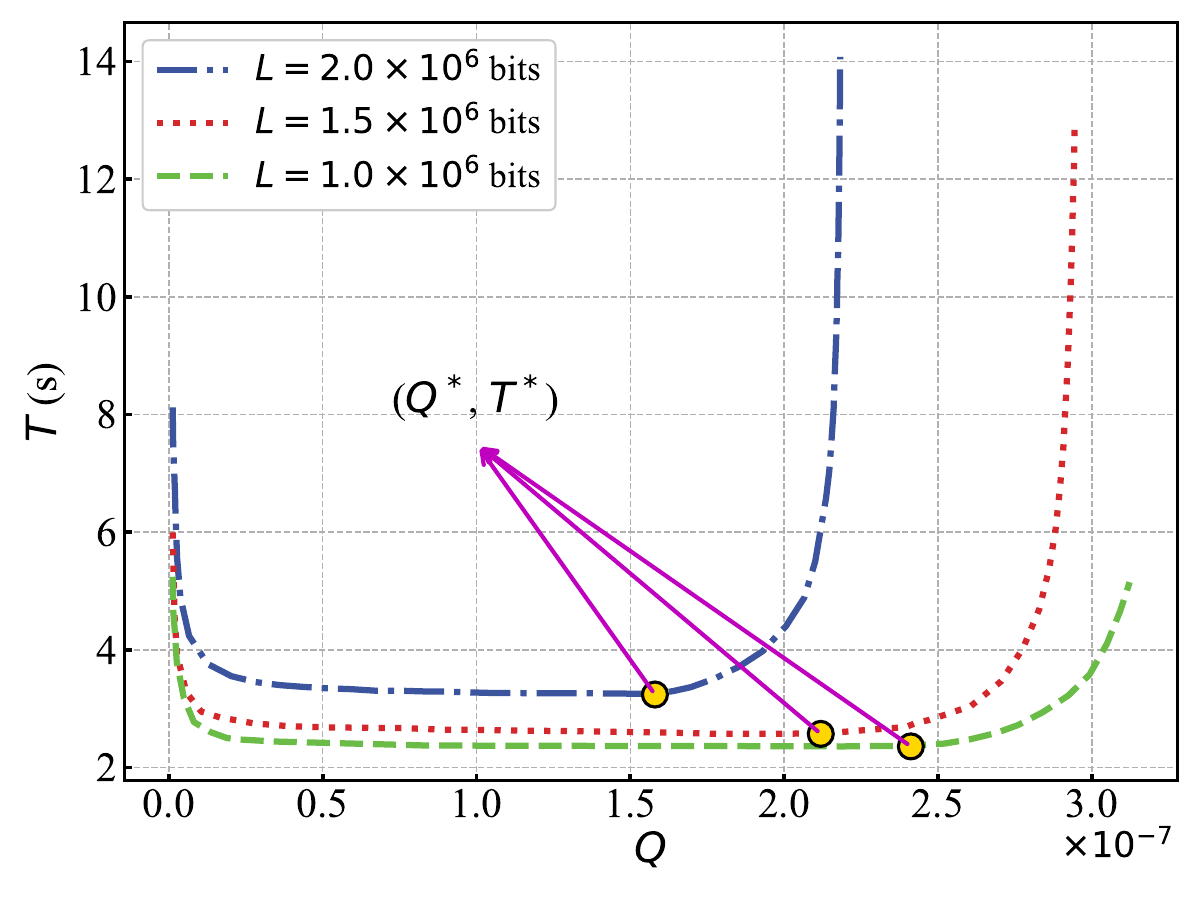}
    \caption{Illustration of Algorithm 2 solving Problem \textbf{P}.}
    \label{E2}
\end{figure}
\par Fig. \ref{E1} and Fig. \ref{E2} illustrate the convergence of the proposed algorithm \ref{alg:P} for solving the problem \textbf{P}. In particular, Fig. \ref{E1} shows that $\{\lambda_i\}_{i\in \mathcal{M}}$ converges rapidly after several iterations at $Q=3\times10^{-8}$ and $Q=3\times10^{-7}$. In the subgradient method we used, adjusting the initial step size led to the appearance of spikes. We observe from Fig. \ref{E2} that we can obtain the optimized $Q$ for minimizing different sizes of uploading local models. In addition, the latency increases as the size of the local model increases. When the size of a local model is fixed, as the value of $Q$ approaches 0, the latency tends to infinity. However, as $Q$ increases, the latency decreases, while as $Q$ exceeds $Q^*$, the latency increases again. In particular, as $Q$ approaches 0, the cooperative jamming is too small, causing the secrecy throughput of the BLDs to approach 0 and the update delay of the DT models to increase. In addition, when $Q$ exceeds $Q^*$, the cooperative jamming on BLDs increases significantly, resulting in a reduction of the available energy for GLDs to upload the local models, thus increasing the upload latency and the total waiting time. This phenomenon shows that the optimal value of $Q$ can be obtained to minimize the FL latency, which is consistent with the theoretical analysis results.
\par For the performance evaluation of the proposed algorithm \ref{alg:P}, a comparison was made between the computation time of CVX's Durobi solver \cite{gurobi2021gurobi} and the enumeration method, with the objective of obtaining identical results. In particular, the test was performed using 10 clusters of local devices in conjunction with a group coordinator. Thanks to the convexity for the subproblems of the problem \textbf{P}, the computational time can be significantly reduced when using the algorithm \ref{alg:P} compared to both the Durobi solver of LINGO and the global solver of LINGO \cite{schrage2006optimization}. As shown in Fig. \ref{E3}, the algorithm \ref{alg:P} is able to achieve identical results as the Durobi solver of CVX and the global solver of LINGO, while requiring significantly less computation time.
\begin{figure}[htbp]
    \centering
    \includegraphics[width=7.8cm]{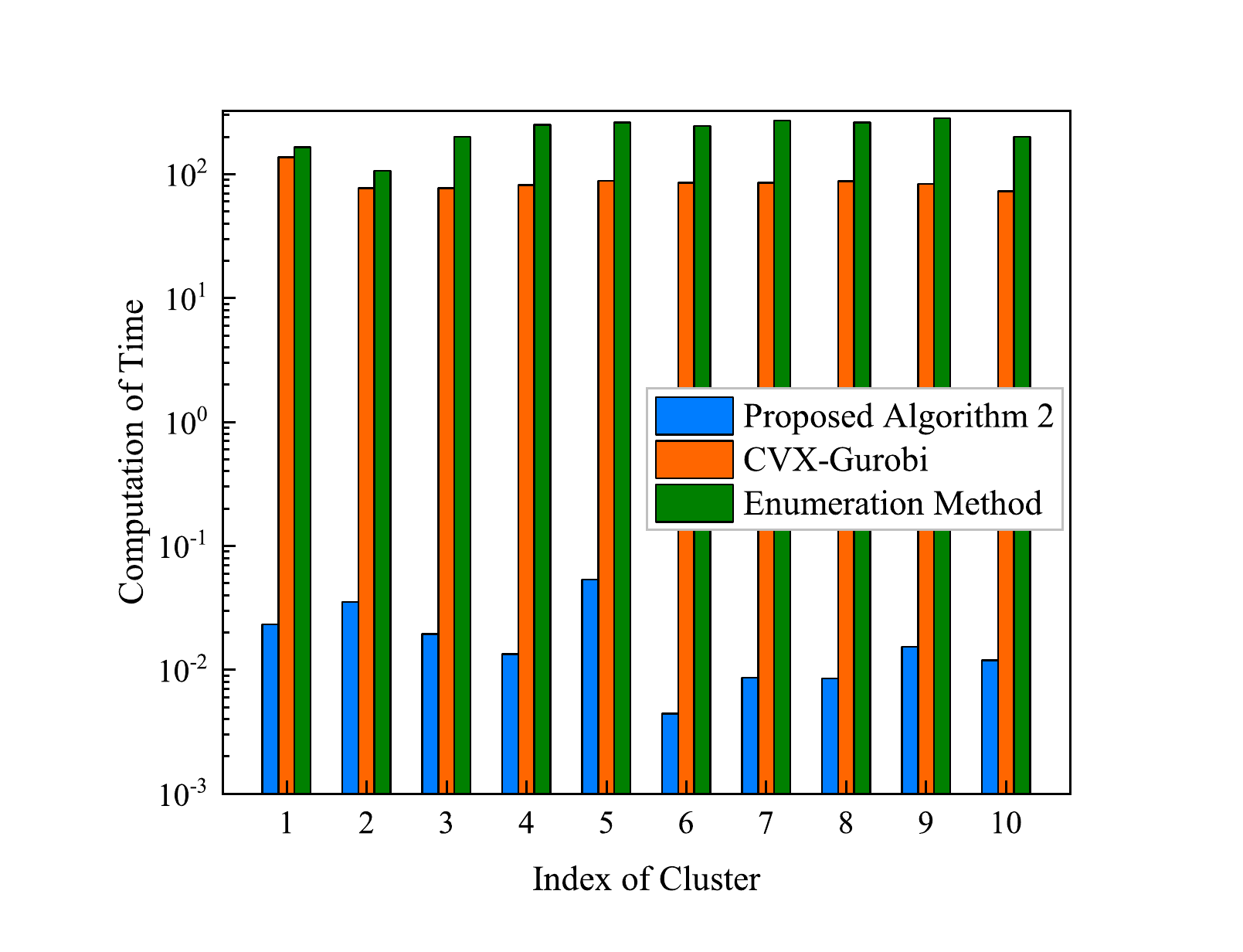}
    \caption{The advantage of Algorithm \ref{alg:P} in computation of time performance.}
    \label{E3}
\end{figure}
\par Fig. \ref{E4} provides further validation of the advantages of our proposed solution in terms of security performance. In particular, the analysis focused on a cluster of 10 local devices linked to a group coordinator. For BLDs DT information updating to the group coordinator, the application of (\ref{eq3}) for the classification and acquisition of four BLDs will increase the security throughput. Furthermore, it will reduce the DT waiting time for synchronisation and the overall FL latency.
\begin{figure}[htbp]
    \centering
    \includegraphics[width=6cm]{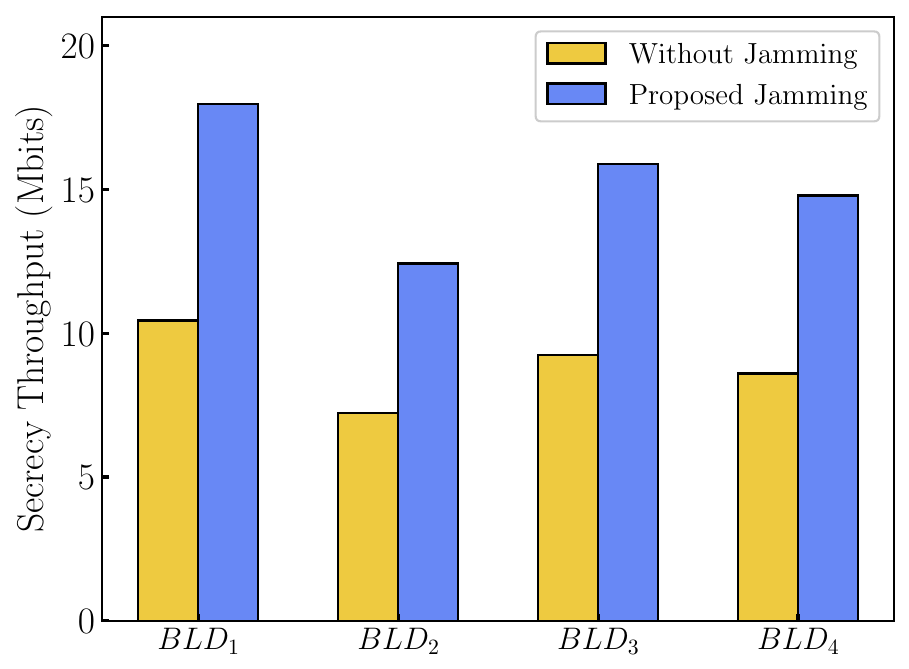}
    \caption{The advantage of proposed cooperative jamming scheme in secure throughput performance.}
    \label{E4}
\end{figure}
\par Fig. \ref{E5} shows the comparison of the global model accuracy of our proposed FL scheme with the benchmarks of FedAVG and a centralized convolutional neural network (CNN), where the FedAVG and the centralized CNN use the MNIST dataset and the CNN model \cite{10190726}. The centralized CNN is trained on the entire dataset, and it is clear that the centralized CNN solution provides the highest accuracy and the fastest convergence time compared to our proposed FL approach and FedAVG. Although the centralized CNN solution is superior to our proposed FL solution, the difference is almost negligible. After 800 seconds of iteration, the accuracy of our solution reaches 97.77\%, which is comparable to that of the centralized CNN solution. It should be noted that although centralized CNN solutions provide better model training performance in terms of accuracy, they perform worse in terms of data security and privacy. Subsequently, FedAVG performs the worst compared to our proposed FL scheme and the centralized CNN solution. Nevertheless, FedAVG achieves an acceptable level of accuracy, about 95.57\%, while ensuring data security and privacy. This is due to the fact that FedAVG protects privacy by training on local data without requiring sharing, and reduces communication overhead by collaborating with multiple terminals and edge devices, which requires longer training times to achieve high accuracy. Since resource-constrained edge devices perform local training on the global model sent by the server, the duration of this local training may increase. On the other hand, compared to FedAVG, our proposed FL scheme uses DTs to perform local training instead of relying on computationally constrained local devices. This approach allows the acquisition of richer and more diverse industrial data, facilitating better learning of different features and patterns, thereby improving both convergence speed and accuracy. Finally, the results presented in Fig. \ref{E5} indicate that incorporating optimal cooperative jamming into our proposed FL scheme accelerates the convergence rate while maintaining comparable levels of accuracy. This is due to the fact that cooperative jamming effectively increases the secrecy throughput of BLDs, ensures synchronized updates corresponding to DTs, and expands the training data available for FL, thus reducing the delay of each FL iteration.
\par To analyze the security gains of our proposed FL solution with respect to blockchain, we configured three malicious nodes to participate in the FL process and compared the global model accuracy between the blockchain-based FL solution \cite{ullah2023verifiable} and the traditional FedAVG without blockchain. As shown in Fig. \ref{Graph4}, the presence of malicious nodes significantly affects the overall performance of FL, resulting in a lower global model accuracy. Our proposed FL scheme achieves a training accuracy of 74.2\%, outperforming both the blockchain-based FL solution and the conventional FedAVG. This improvement can be attributed to our use of blockchain and the validator selection algorithm, which ensures the integrity verification of participants and the trained model, thereby mitigating the influence of malicious nodes on the final aggregation of the global model.

\begin{figure}[htbp]
    \centering
    \includegraphics[width=7.6cm]{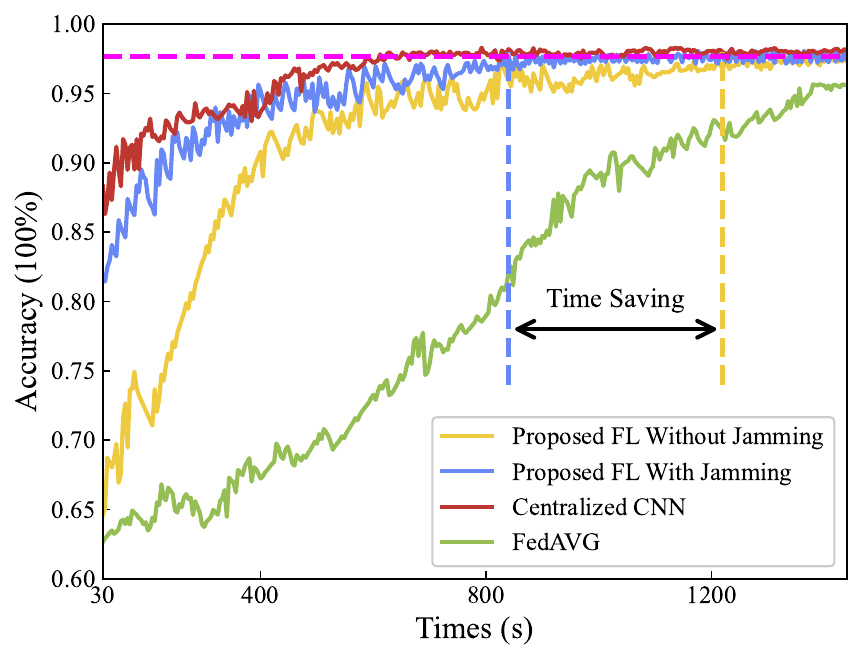}
    \caption{Comparison of accuracy of FL global model.}
    \label{E5}
\end{figure}

\begin{figure}[htbp]
    \centering
    \includegraphics[width=6.5cm]{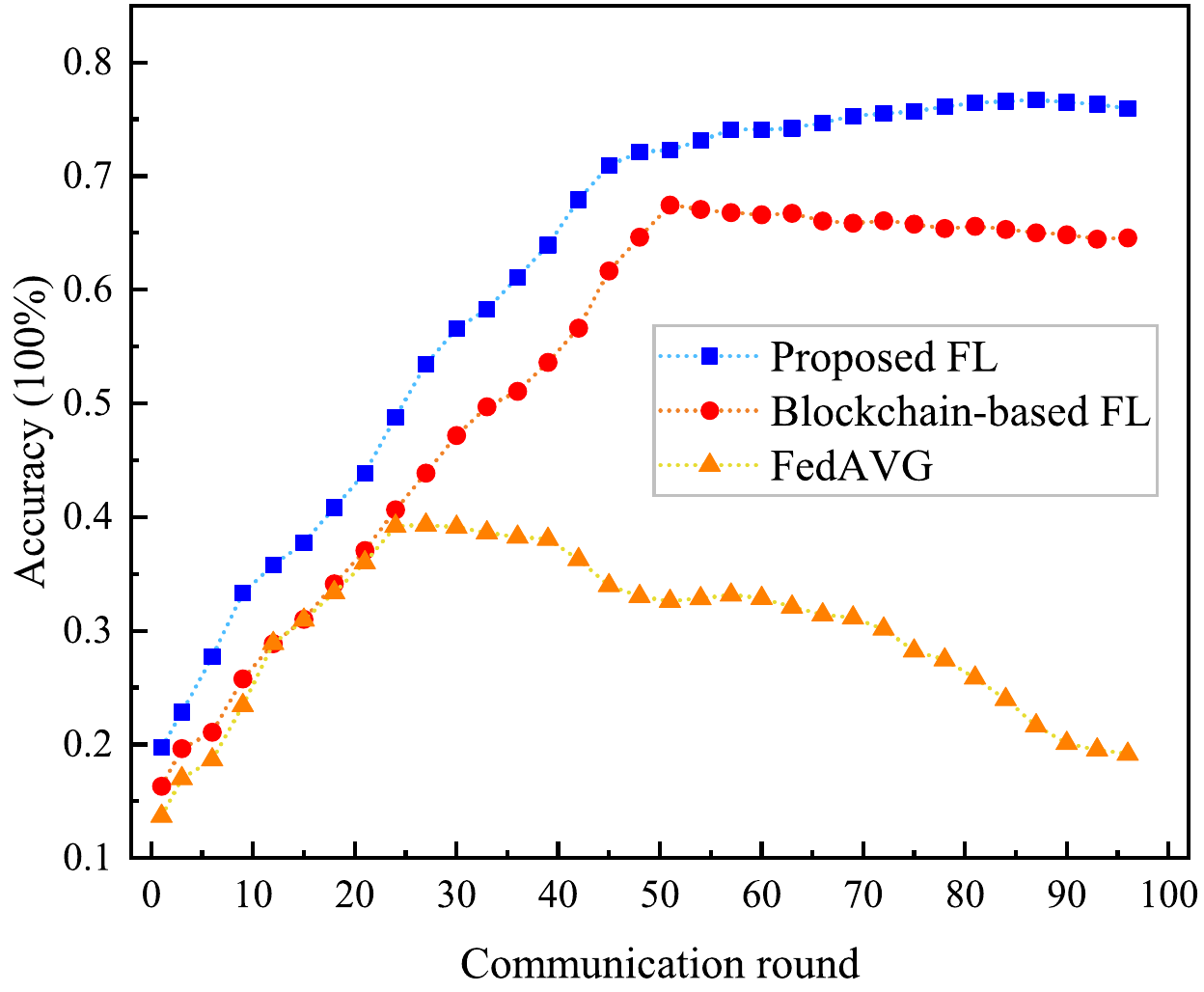}
    \caption{Global model accuracy with malicious nodes.}
    \label{Graph4}
\end{figure}

\begin{figure}[htbp]
    \centering
    \includegraphics[width=7.5cm]{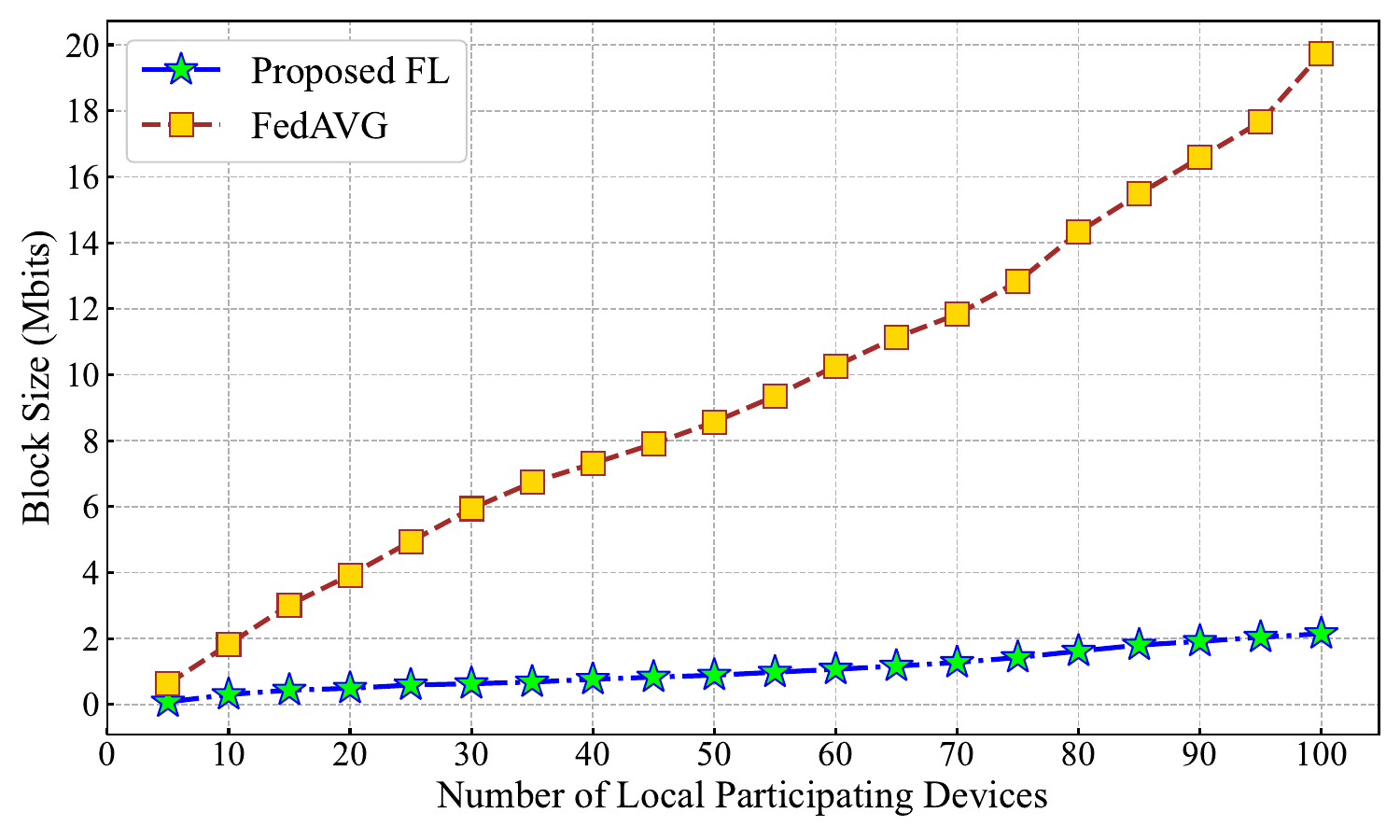}
    \caption{Block size comparison between the proposed FL and FedAVG.}
    \label{E6}
\end{figure}
\par Compared to the conventional FedAVG solution, our proposed FL solution demonstrates superior performance in terms of blockchain resource utilization. Fig. \ref{E6} illustrates the number of local participating devices versus block size for both the proposed FL and FedAVG. It is obvious that the block size for our proposed FL remains relatively small compared to the traditional FedAVG solution. In essence, the number and size of blocks to be added to the blockchain depend on the number of participants, and the number of group coordinators, as well as the number of validators. In a FedAVG context with blockchain, the number of blocks will be equal to the number of participants involved in the learning process. Consequently, the size of the final block, which contains all the aggregated local models, will increase significantly. In our proposed FL scheme, the group coordinator first performs a joint averaging of the local models within the aggregation cluster, after which the validator generates the aggregation model. Both the number and size of blocks added to the blockchain are significantly reduced, which consequently reduces the resources occupied by the blockchain in the cloud tier. 

\section{Conclusion}
In this paper, we present a DT and blockchain-assisted FL scheme to address the security and privacy issues associated with resource-constrained IIoT devices participating in FL. In our proposed approach, resource-constrained edge devices can generate DTs through group coordinators to perform local training for FL, which is then aggregated with the locally trained models from resource-rich devices. We formulated an FL delay minimization problem that includes the upload transmission time of capable local devices, the DT secure synchronization time, the processing rates of capable local devices, the upload power, and the interference from malicious attackers. Then, we proposed a cooperative jamming optimization algorithm to solve this problem. Specifically, we introduced upper bounds on the local device training delay and the effects of aggregation jamming as auxiliary variables, transforming the problem into a convex optimization problem that can be decomposed for independent solution. In addition, we employed a blockchain verification mechanism to ensure the integrity of uploaded models and participant identities throughout the FL process. Simulation results have verified the effectiveness of our proposed cooperative jamming FL delay optimization algorithm, demonstrating that the DT and blockchain-assisted FL scheme significantly outperforms benchmark schemes in terms of execution time, block optimization, and accuracy. In future work, we will investigate the impact of manipulated DT synchronization models on the accuracy of the final model, and explore the potential of combining edge AI to detect malicious nodes.


%

\appendices
\section{}
Let $f_i(t^{\rm up}_{\rm G})$ denote the second term in constraint (\ref{eq24a}), which is represented as follows:
\begin{equation}\label{eq37}
     f_i(t^{\rm up}_{\rm G})=\frac{n_{\rm C}}{g_{i{\rm C}}}(2^{\frac{L}{Wt^{\rm up}_{\rm G}}}-1)2^{\frac{L}{Wt^{\rm up}_{\rm G}}}, \forall i\in \mathcal{M}.
\end{equation}
Then, we compute the second-order derivative of $f_i(t^{\rm up}_{\rm G})$ with respect to $t^{\rm up}_{\rm G}$, which can be represented as follows: 
\begin{equation}\label{eq38}
\begin{aligned}
    \frac{\partial ^{2} f_i(t^{\rm up}_{\rm G})}{\partial {(t^{\rm up}_{\rm G})}^{2}}=& \frac{(\ln{2})^2n_{\rm C}L^22^{\frac{L}{Wt^{\rm up}_{\rm G}}}}{g_{i{\rm C}}W^2(t^{\rm up}_{\rm G})^3}\\
&\times[2^{\frac{L}{Wt^{\rm up}_{\rm G}}}+(i-1)^2(2^{\frac{(i-1)L}{Wt^{\rm up}_{\rm G}}}-1)\\
& +(i-1)2^{1+\frac{L}{Wt^{\rm up}_{\rm G}}}]>0, \forall i\in \mathcal{M}.
\end{aligned}
\end{equation}

According to (\ref{eq38}), the second-order derivative of $f_i(t^{\rm up}_{\rm G})$ with respect to $t^{\rm up}_{\rm G}$ is always positive. Therefore, $f_i(t^{\rm up}_{\rm G})$ is a convex function with respect to $t^{\rm up}_{\rm G}$. Moreover, since the first and the third terms in constraint (\ref{eq24a}) is convex with respect to $y$ and $\{q_i\}_{i\in \mathcal{M}}$, the left hand side of constraint (\ref{eq24a}) is a convex function. In addition, constraints (\ref{eq24b})-(\ref{eq24d}) are all affine functions. Thus, problem \textbf{P}-GLD is a convex optimization problem. 

\section{}
The value of $\mu$ was set to be equal to one of the sets $\{\frac{\lambda_i\widehat{t}^{\rm up}_{\rm B}}{g_{i{\rm E}}}\}_{i\in \mathcal{M}}$. To demonstrate this, we assume that there are some $\widehat{i}\in \mathcal{M}$ and $\mu=\frac{\lambda_{\widehat{i}}\widehat{t}^{\rm up}_{\rm B}}{g_{i{\rm E}}}$. Consequently, the values of $\widehat{i}$ and $\mu$ permit the derivation of the following conclusion.

\begin{equation}\label{eq39}
\lambda_i\widehat{t}^{\rm up}_{\rm B}-\mu g_{i{\rm E}}=
\begin{cases}
\ge 0, & map(i)\leq map(\widehat{i})-1 \\
=0, & i=\widehat{i}\\
\leq 0, & map(i)\ge map(\widehat{i})+1
\end{cases}.
\end{equation}

As indicated in (\ref{eq34}) and (\ref{eq39}), the optimal solution for $\{q_i\}_{i\in \mathcal{M}}$ can be derived from (\ref{eq35}). It should be noted that the value of $q_i$ does not affect the result when $i=\widehat{i}$ is given. Nevertheless, as a consequence of constraint (\ref{eq24b}), the value of $q_{\widehat{i}}$ is restricted to $\frac{Q-\sum_{map(r)\ge map(\widehat{i})+1}q_rg_{r{\rm E}}}{g_{\widehat{i}{\rm E}}}$.
\par Next, we prove the uniqueness of 
$\widehat{i}$. Suppose there exists 
$\widehat{i}=i^*$ such that the following equation holds:
\begin{equation}\label{eq40}
0\leq \frac{Q-\sum_{map(r)\ge map(i^*)+1}q_rg_{r{\rm E}}}{g_{i^*{\rm E}}}\leq Q^{\rm max}_{i^*}.
\end{equation}
Then we further assume that 

\begin{equation}\label{eq41}
\widehat{i}=i\in \mathcal{M}|map(i)=map(i^*)+m.
\end{equation}

\par If $m\ge 1$, we have
\begin{equation}\label{eq42}
\begin{aligned}
&\frac{Q-\sum_{map(r)\ge map(\widehat{i})+1}q_rg_{r{\rm E}}}{g_{\widehat{i}E}}\\&=\frac{1}{g_{\widehat{i}{\rm E}}}(Q-\sum_{map(r)\ge map(i^*)+1}Q_r^{\rm max}g_{r{\rm E}}+Q_{\widehat{i}}^{\rm max}g_{\widehat{i}{\rm E}}\\&+\sum_{map(i^*)+m+1\ge map(r)\ge map(i^*)+1}Q_r^{\rm max}g_{r{\rm E}})\\&\ge \frac{Q-\sum_{map(r)\ge map(i^*)+1}Q_r^{\rm max}g_{r{\rm E}}+Q_{\widehat{i}}^{\rm max}g_{\widehat{i}{\rm E}}}{g_{\widehat{i}{\rm E}}}\\&\ge\frac{0+Q_{\widehat{i}}^{\rm max}g_{\widehat{i}{\rm E}}}{g_{\widehat{i}{\rm {\rm E}}}}=Q_{\widehat{i}}^{\rm max}.
\end{aligned}
\end{equation}
Similar to (\ref{eq42}), when $m \leq -1$, the following inequality can be obtained:
\begin{equation}\label{eq43}
\frac{Q-\sum_{map(r)\ge map(\widehat{i})+1}q_rg_{r{\rm E}}}{g_{\widehat{i}{\rm E}}}\leq0.
\end{equation}

\par Based on (\ref{eq42}) and (\ref{eq43}), it can be concluded that no $\widehat{i}\neq i^*$ exists which satisfies (\ref{eq40}). So it shows that the value of $\widehat{i}$ is unique.


\ifCLASSOPTIONcaptionsoff
  \newpage
\fi



%
\bibliographystyle{IEEEtran} 
\bibliography{IEEEabrv,ref} 

%




\end{document}